\documentclass{article}

\usepackage[dvipsnames]{xcolor}
\usepackage{hyperref}

\usepackage{amsmath,amsthm,amssymb,natbib,verbatim,cancel, bbm, tikz, tkz-euclide,graphicx,subcaption}

\usepackage{caption}
\usetikzlibrary{bending}
\usetikzlibrary{shapes.misc, positioning, arrows,decorations.markings}
 \usetikzlibrary{calc}
 \tikzset{
  prefix after node/.style={prefix after command=\pgfextra{#1}},
  /semifill/ang/.initial=45,
  /semifill/upper/.initial=none,
  /semifill/lower/.initial=none,
  semifill/.style={
    circle, draw,
    prefix after node={
      \pgfqkeys{/semifill}{#1}
      \path let \p1 = (\tikzlastnode.north), \p2 = (\tikzlastnode.center),
                \n1 = {\y1-\y2} in [radius=\n1]
            (\tikzlastnode.\pgfkeysvalueof{/semifill/ang}) 
            edge[
              draw=none,
              fill=\pgfkeysvalueof{/semifill/upper}, fill opacity = .35,
              to path={
                arc[start angle=\pgfkeysvalueof{/semifill/ang}, delta angle=180]
                -- cycle}] ()
            (\tikzlastnode.\pgfkeysvalueof{/semifill/ang}) 
            edge[
              draw=none,
              fill=\pgfkeysvalueof{/semifill/lower}, fill opacity = .35,
              to path={
                arc[start angle=\pgfkeysvalueof{/semifill/ang}, delta angle=-180]
                -- cycle}] ();}}}

\hypersetup{
    colorlinks =true,
    allcolors =NavyBlue,
    allbordercolors =white,
}

	\theoremstyle{plain}
	\newtheorem{theorem}{Theorem}

	\newtheorem*{axiom*}{Axiom}

	\newtheorem{lemma}{Lemma}
	
	\newtheorem{corollary}{Corollary}

	\theoremstyle{definition}
	
	\newtheorem{example}{Example}

	\theoremstyle{remark}

	\newtheorem*{claim*}{Claim}
	
\title{Revealed Invariant Preference\footnote{We would like to thank Federico Echenique, Faruk Gul, Daniel Luo, Efe Ok, Larry Samuelson, Chris Shannon, Leeat Yariv, as well as seminar audiences at NASMES '24, D-TEA '24, UC Berkeley, Caltech, and Princeton for helpful comments. Peter Caradonna thanks the Linde Institute for financial support and Princeton University for their hospitality, where part of this work was completed. We also thank Ben Wincelberg for excellent research assistance. This paper was previously presented under the title `Model-Based Revealed Preference.'}}
\author{Peter Caradonna and Christopher P. Chambers}

\begin{document}
\maketitle

\begin{abstract}
    We consider the problem of rationalizing choice data by a preference satisfying an arbitrary collection of \emph{invariance} axioms.  Examples of such axioms include quasilinearity, homotheticity, independence-type axioms for mixture spaces, constant relative/absolute risk and ambiguity aversion axioms, stationarity for dated rewards or consumption streams, separability, and many others. We provide necessary and sufficient conditions for invariant rationalizability via a novel approach which relies on tools from the theoretical computer science literature on automated theorem proving.   We also establish a generalization of the Dushnik-Miller theorem, which we use to give a complete description of the out-of-sample predictions generated by the data under any such collection of axioms. 
\end{abstract}

\medskip


\section{Introduction}

Nearly all economic models restrict, in some manner, the class of preferences that agents  hold. When these restrictions are at odds with the broad, empirical regularities in how individuals actually evaluate various trade-offs and decisions, this misspecification may introduce errors which bleed into other aspects of the model, leading to unrealistic or even outright incorrect predictions (e.g.\ \citealt{mehra1985equity}).\medskip

This motivates a basic need to be able to obtain, in systematic fashion, the empirical content of the wide variety of additional assumptions on behavior, beyond rationality alone, which are typically imposed in practice. However, in many cases, characterizations of the testable implications of these extra assumptions exist only for restrictive, special types of data sets (e.g.\ price-consumption data), or remain wholly unknown. Where results do exist, they are often obtained through model-specific considerations that leave them unable to be straightforwardly adapted to apply to other, even closely related sets of assumptions.\footnote{E.g. the techniques of \cite{afriat1967construction}, which rely crucially on the the functional forms of particular representations.}\medskip

In this paper we provide a complete characterization the testable implications of a wide variety of common preference and decision axioms, through a novel, general approach.  Our results apply to arbitrary revealed preference data sets, and require no assumptions on the domain of choice. To illustrate our approach, the following example considers the problem of finding a time-stationary rationalizing preference for simple choice data. It highlights how the additional structure imposed by stationarity, above and beyond rationality alone, complicates the problem of evaluating consistency.

\begin{example}\label{stationaritycex}
    Suppose that in each period, an agent is able to consume a fruit of their choice from the set $F= \{\textrm{Apple}, \textrm{Banana}, \textrm{Cherries}, \textrm{Dragonfruit}\}$, and suppose that we observe an individual's choice behavior between various discrete-time, infinite-horizon consumption streams.\medskip
    
    We are interested in whether the individual's choice behavior is consistent with the maximization of a rational preference that is additionally \emph{stationary}, in the sense of \cite{koopmans1960}. A preference over consumption streams is said to be stationary if, for any pair of streams $x$ and $y$:
    \[
        (x_1, x_2, \ldots ) \succeq (y_1, y_2, \ldots) \quad \iff \quad (f, x_1, \ldots ) \succeq (f, y_1, \ldots),
    \]
    where $f$ denotes any fruit in $F$.  In other words, a preference is stationary if, whenever we take two streams, delay them each by one period and insert a common item in each first period, the preference between them does not reverse.\medskip

    Suppose now that, for two fixed streams $x$ and $y$, we observe the following revealed preference data:
    \begin{equation}\label{stationaritycex1}
        \begin{aligned}
            (a,x_1,\ldots) & \succ^R (b, y_1, \ldots)\\
            (b,x_1,\ldots) & \succ^R (a, y_1, \ldots),
        \end{aligned}
    \end{equation}
    and
    \begin{equation}\label{stationaritycex2}
        \begin{aligned}
            (c,y_1,\ldots) & \succ^R (d, x_1, \ldots)\\
            (d,y_1,\ldots) & \succ^R (c, x_1, \ldots).
        \end{aligned}
    \end{equation}
    This data set is wholly consistent with rational behavior. Indeed, the revealed preference itself is a transitive binary relation, which more than suffices to ensure its consistency with the paradigm of rational, optimizing behavior.\footnote{See, e.g., \cite{richter1966revealed} Theorem 1, or \cite{chambers2016revealed} Theorem 2.6 for a textbook treatment.}\medskip

    On the other hand, the data are inconsistent with \emph{any} stationary preference. To see this, suppose $\succeq$ is a preference relation that agrees with the observed comparisons \eqref{stationaritycex1} and \eqref{stationaritycex2}. Any such preference must specify a ranking between $x$ and $y$. If $y \succeq x$, the observations in \eqref{stationaritycex1} would imply, under stationarity, that:
    \[
        (a,x_1, \ldots) \succ \underbrace{(b, y_1, \ldots) \succeq (b, x_1, \ldots)}_{\textrm{By stationarity}} \succ \overbrace{(a, y_1, \ldots) \succeq (a, x_1, \ldots)}^{\textrm{By stationarity}},
    \]
    and hence that $\succeq$ was in fact not transitive. However, an identical argument applied to the observations in \eqref{stationaritycex2} yields that no stationary preference which agrees with the data can rank $x \succeq y$ either. 
    \hfill $\blacksquare$
\end{example}

The subtlety in testing for stationarity in \autoref{stationaritycex} arose from the interdependence it imposed between rankings over related pairs of consumption streams. Perhaps surprisingly, a wide range of decision-theoretic axioms introduce precisely the same type of interdependency.\medskip

Suppose now that $X$ is an abstract consumption space, and let $\mathcal{M}$ denote a collection of transformations, each mapping $X \to X$.  We say that a preference $\succeq$ on $X$ is \emph{invariant} under a transformation $\omega \in \mathcal{M}$ if:
\[
    x \succeq y \quad \iff \quad \omega(x) \succeq \omega(y),
\]
for every $x,y \in X$.  An \emph{invariance axiom} is then simply the requirement that a preference be invariant under every transformation in some collection $\mathcal{M}$.  In \autoref{stationaritycex}, these transformations were the maps 
\[
(x_1,x_2,\ldots) \mapsto (f, x_1, \ldots),
\]
for each fruit $f \in F$. However, many other axioms of first-order economic importance are also of this form. For example, the independence axiom of \cite{von1947theory} is an invariance axiom: there, $X$ is a lottery space and $\mathcal{M}$ consists of all transformations of the form:
\[
    p \mapsto \alpha p + (1-\alpha) q
\]
for some $\alpha \in (0,1]$ and lottery $q$. But, by varying our choice of $X$ and $\mathcal{M}$, we also obtain the standard quasilinearity, homotheticity axioms in consumer theory, constant absolute or relative risk aversion, and many others as special cases.\medskip



More formally, in this paper we provide complete answers to the two following questions: 
\begin{itemize}
    \item[(Q.1)]\hypertarget{q1}{} When is a given revealed preference data set consistent with the maximization of some preference relation that satisfies an arbitrary collection of invariance axioms?
    \item[(Q.2)]\hypertarget{q2}{} What comparisons \emph{not} observed in the data are nonetheless agreed upon by \emph{every} invariant rationalizing preference?\footnote{Conditional upon the set of such preferences being non-empty.} 
\end{itemize}

\noindent \autoref{stationaritycex} highlights that even sets of comparisons that are very sparse can lead to strong (sometimes even impossible-to-fulfill) restrictions on the possible comparisons a consistent, invariant preference can make.\medskip

The source of this complexity is that the addition of a new comparison between unranked pairs of alternatives necessarily also fixes the comparisons between each image of this pair, under every transformation in $\mathcal{M}$. We term these additional implications the \emph{knock-on effects} of the addition.  As illustrated in \autoref{stationaritycex}, these extra implications may form transitivity violations, even when the initial addition itself does not.\medskip

In order to account for the potential infinity of knock-on effects, we are forced to consider \emph{sets} of simultaneous restrictions on the possible comparisons a rationalizing preference can make. In turn, these sets of restrictions can be combined to deduce further constraints which may not emerge directly from the data. \medskip

We introduce a simple, binary operation that we term the `collapse,' that converts a suitable pair of restriction sets into a new one.  This operation may roughly be viewed as a set-valued analogue of the act of deducing $x \succeq z$ from a pair of compatible observations $x \succeq y$ and $y \succeq z$ via transitivity.  We show that, no matter the complexity of the environment or structure of the family of transformations, a simple no-cycle condition, phrased in terms of our collapse operation, fully characterizes rationalizability by an invariant preference.  We also prove that a related generalization of the transitive closure, again in terms of our collapse operation, completely characterizes the set of out-of-sample predictions generated by the data under any given set of invariance axioms.\medskip

Our methodology relies intimately on a connection with formal logic. We first recast the problem of finding a consistent, invariant preference as one of testing the satisfiability of a set of clauses. We establish that a `cycle' in our sense may be used to construct a formal proof of unsatisfiability in the accompanying logical system. To prove the converse, we utilize a result due to \cite{robinson1965machine}, which establishes that if a given system of clauses is unsatisfiable, there exists a proof of this fact with specific combinatorial structure. We then show that any proof of unsatisfiability, of the precise form guaranteed by Robinson's theorem, can always be `inverted' to construct a cycle in our original sense.\medskip

The paper proceeds as follows. In Section 2 we formally state our problem and provide a number of examples of economic axioms covered by our results.  Section 3 considers a special case of our general result---the case in which all the transformations defining our invariance axiom commute. In this special case, we show that our general no-cycle condition reduces to a particularly simple form. In Section 4, we introduce our collapse operation, and provide our general characterization of invariant rationalizability. In Section 5 we consider an extension in which we the transformations in $\mathcal{M}$ are instead \emph{partial} functions, only defined on subsets of $X$.  We show that, in this broader setting, all of the results of Section 4 remain true as stated, and provide a number of new economic applications. Finally, in Section 6 we provide a novel generalization of the Dushnik-Miller theorem (\citealt{dushnik1941}) for invariant preferences, which we use to characterize the set of out-of-sample, or counterfactual, predictions generated by a the data and a set of axioms. Section 7 concludes.


\subsection{Related Literature}

The revealed preference literature is too large to adequately survey here, see \cite{chambers2016revealed} for an overview.\footnote{See also \cite{echenique2020new} for a summary of some recent work in this space.}  Classically, \cite{richter1966revealed} was the first to characterize rationalizability for the abstract choice model. We obtain Richter's original theorem as a special case of our main results (see Section 3).  A similarly classic reference in this vein is \cite{duggan} who retains an abstract framework but imposes additional restrictions on the interpretation of `rationality.'\medskip

Other authors have studied the problem of rationalizing choice data via preferences with various general structures. \cite{nishimura2017comprehensive} study this problem for continuous and monotone preferences on various spaces. \citet{demuynck2009} investigates a general class of `closure operators' on spaces of binary relations that generalize the transitive closure, and obtains a general extension result for algebraic structures satisfying certain properties.\footnote{See \citet{ward} for a general theory of closures.}  While general, applying these tools requires non-trivial effort to establish their conditions are satisfied.  In contrast, our results focuses on a smaller mathematical class of algebraic properties, invariance and monotonicity axioms, but are able to derive results that are immediately applicable.\medskip

Other authors have considered invariant preferences in various contexts.  \cite{ok2014topological,ok2021fully} consider various extension results for invariant preorders on groups. In contrast, we consider both a more general class of primitive relations and more general notion of invariance.\footnote{Mathematically, our notion of invariance corresponds to invariance of a preference under an arbitrary semi-group action on the consumption space. For definitions, see \cite{fuchs2011partially}.} Recently \citet{freer2022}, building off the tools of \cite{demuynck2009}, consider the problem of invariant rationalization by incomplete or non-transitive binary relation.\footnote{They also establish an invariant rationalizability result in the special case the collection of transformations, under composition, forms a linearly ordered group.} \cite{dubra2004} show that every `incomplete' expected utility (EU) preference may be completed in such a way as to preserve the EU axioms.\medskip

\cite{dushnik1941} show that every partial order is equal to the intersection of its linear order extensions.  Several authors in economics have taken interest in such unanimity, or Pareto, representation of incomplete preferences.  Abstract approaches include \cite{donaldson1998,bossert1999,weymark2000generalization} and \cite{alcantud2009}. In concrete economic environments, similar representations can be found in, for example, the theory of expected utility preferences (\citealt{dubra2004,gorno2017strict}), Krepsian style  preferences over menus (\citealt{nehring1999multi}), or rankings of accomplishments (\citealt{chambersmiller2018}).\medskip

We are not the first paper to exploit the connection between revealed preference and formal logic. \cite{chambers2014axiomatic} study the general form of empirical content for theories in first-order logic, relating the syntax of first-order theories to the empirical content via a type of sentence they call ``UNCAF'' (universal negation of conjunctions of atomic formulae).   \citet{chambers2017general} establishes that theoretical relations in theories axiomatizable by universal sentences can be eliminated, resulting in a theory which is itself universally axiomatizable. Such an axiomatization results from enumerating all logical consequences of the original theory without theoretical relations.  Thus, these two papers give a r.e. method which could in principle enumerate datasets which are inconsistent with a given theory.\footnote{See also \citet{chambers2016revealed}, Chapter 13.}  In comparison, our results rely only on the simpler framework of propositional logic, and provide a more practical method for understanding inconsistent data. \cite{gonczarowski2019infinity} show that similar connections with propositional logic obtain in a variety of economic contexts. \cite{galambos2019descriptive,yildiz2020regularities} investigate the relation between computational complexity of revealed preference theories and their logical syntax.  \medskip

\cite{robinson1965machine} showed that a certain algorithmic operation on logical clauses called resolution was sound and refutation-complete. This reduced the problem of proving a set of clauses to be inconsistent without constructing a truth table to a discrete search problem. A number of extensions and refinements giving various `normal forms' for proofs were established in the early artificial intelligence literature to attempt to further reduce the complexity of this search space (see, e.g., \citealt{schoning2008logic} for an overview).\medskip

Finally, our work presupposes no notion of topology, but many works in economics consider topological aspects of the extension problem.  \cite{aumann1962utility,aumann1964utility,peleg1970,levin} are classical references, but the theory has developed much since then (e.g., \citealt{ok2002, nishimura2017comprehensive}).

\section{The Model}

Let $X$ denote set of alternatives. A {\bf preference relation} $\succeq$ is a complete and transitive binary relation on $X$. Given a preference relation, we use $\succ$ and $\sim$ to denote its asymmetric and symmetric components, respectively.\medskip

Let $\mathcal{M}$ denote a set of transformations, each mapping $X \to X$. We say that a preference relation is $\mathcal{M}$-{\bf invariant} if, for all $x,y \in X$ and all $\omega \in \mathcal{M}$:
\begin{equation}\label{invdef}
    x \succeq y \quad \implies \quad \omega(x) \succeq \omega(y),
\end{equation}
and
\begin{equation}\tag{1}
    x \succ y \quad \implies \quad \omega(x) \succ \omega(y).
\end{equation}
Note that if $\omega, \omega' \in \mathcal{M}$, then any $\mathcal{M}$-invariant preference also satisfies:
\[
    x \succeq y \; \iff \; (\omega \circ \omega')(x) \succeq (\omega \circ \omega')(y) \; \iff \; (\omega' \circ \omega)(x) \succeq (\omega' \circ \omega)(y),
\]
and analogously for strict preferences. As such, without loss of generality we will assume (i) the identity function $\textrm{id} \in \mathcal{M}$, and (ii) $\mathcal{M}$ is closed under composition.\footnote{Formally, we assume, without loss of generality, that $(\mathcal{M}, \circ)$ forms a semigroup with identity, or a \emph{monoid}.} \medskip

We consider a pair of observed {\bf revealed preference} relations as data, denoted $\langle \succsim^R, \succ^R \rangle$, where $\succ^R$ is a sub-relation of $\succsim^R$.\footnote{We do not, however, assume that $\succ^R$ is necessarily the asymmetric component of $\succsim^R$.}  We call this tuple an \emph{order pair}.  These relations could arise through observed choice behavior, e.g.\ by defining:
\begin{itemize}
    \item $x \succsim^R y$ if $x$ and $y$ belong to a common choice set, from which it was observed $x$ was chosen.
    \item $x \succ^R y$ if $x \succsim^R y$ and, in addition, $y$ was not chosen.
\end{itemize}
However, we explicitly allow for them also encoding other salient properties such as \emph{monotonicity restrictions}, by setting $x \succ^R y$ if $x$ dominates $y$ in a particular partial order of interest. We will assume, without loss of generality, that $\succsim^R$ is reflexive.\medskip

We interpret $\langle \succsim^R, \succ^R\rangle$ as {\bf data}, and seek to understand when it is consistent with the behavior of an economic actor who chooses to maximize some $\mathcal{M}$-invariant preference relation $\succeq$.  Formally, an $\mathcal{M}$-invariant preference relation $\succeq$ {\bf rationalizes} the data $\langle \succsim^R, \succ^R\rangle$ if both (i) $\succsim^R \; \subseteq \; \succeq$, and (ii) $\succ^R \; \subseteq \; \succ$. The primary result of our paper will be to provide a complete characterization of those data sets $\langle \succsim^R ,\succ^R\rangle$ that are rationalizable by an $\mathcal{M}$-invariant preference, for any choice of $X$ and $\mathcal{M}$.  Equivalently, we characterize which data sets may \emph{not} be rationalized by any $\mathcal{M}$-invariant preference (for fixed choice of $\mathcal{M}$), which therefore provides a complete description of the empirical content of such models.\medskip

The transitive closure of our pair $\langle \succsim^R , \succ^R\rangle$ is the pair of relations $\langle \succsim^R_\intercal, \succ^R_\intercal\rangle$, where $x \succsim_\intercal^R y$ if and only if there exists some finite sequence $x_0, \ldots, x_N \in X$ such that:
\[
    x = x_0 \succsim^R x_1 \succsim^R \cdots \succsim^R x_N = y.
\]
Similarly, $x \succ^R_\intercal y$ if $x \succsim^R_\intercal y$ and some relation in the sequence belongs to $\succ^R$. A pair of relations $\langle \succsim^R ,\succ^R\rangle $ is said to be {\bf acyclic} if there do not exist $x_0, \ldots, x_N \in X$ such that:
\[
    x_0 \succsim^R x_1 \succsim^R \cdots \succsim^R x_N \succ^R x_0.
\]
We refer to such a sequence as a {\bf cycle}.  

\subsection{Examples of Invariant Preferences}

In this section we provide a number of examples of well-known invariance axioms. All of these correspond to various special cases of our notion of $\mathcal{M}$-invariance, for particular choices of $X$ and $\mathcal{M}$.   

\subsubsection{Quasilinearity}

Let $X = \mathbb{R}_+ \times Z$. A preference is said to be \emph{quasilinear} if:
\[
    (t, z) \succeq (t',z') \quad \iff \quad (t+ \alpha, z) \succeq (t' + \alpha , z')
\]
for all $(t,z), (t',z') \in X$ and $\alpha \ge 0$. When $(t,z)$ is interpreted as the `dated reward,' corresponding to the delivery of a prize $z$ to an agent at $t$ units of time into the future, quasilinearity is also referred to as \emph{stationarity} (see \citealt{fishburn1982}).  See also the notion of `$\phi$-additivity' in \cite{caradonna2023preference}.

\subsubsection{(Generalized) Homotheticity}

Let $X$ be a cone in a real vector space. A preference is \emph{homothetic} if:
\[
    x \succeq y \quad \iff \quad t x \succeq t y
\]
for all $x,y\in X$ and all scalars $t > 0$. Similarly, Cobb-Douglas preferences are the unique, continuous and monotone preferences on $\mathbb{R}^N_+$ satisfying the related but more general form of invariance:
\[
    (x_1,\ldots, x_N) \succeq (y_1,\ldots, y_N) \quad \iff \quad (t_1 x_1, \ldots, t_N x_N) \succeq (t_1 y_1, \ldots, t_N y_N),
\]
for all $x,y \in \mathbb{R}^N_+$ and all $(t_1,\ldots, t_N) \in \mathbb{R}^N_{++}$ (see \citealt{trockel1989classification}).

\subsubsection{Mixture Invariance}

Suppose that $X = \Delta(Z)$, the set of all Borel probability measures on some metrizable space $Z$. A preference satisfies the \emph{independence} axiom of von Neumman and Morgenstern (\citealt{von1947theory}) if:
\[
    \mu \succeq \nu \quad \iff \quad \alpha \mu + (1-\alpha) \eta \succeq \alpha \nu + (1-\alpha) \eta
\]
for all $\alpha \in (0,1]$ and $\eta \in X$.\footnote{More generally, this form of invariance may be defined for any mixture space. See, e.g., \cite{herstein1953axiomatic}.} If instead $X$ denotes the Anscombe-Aumann domain $\mathcal{F}$ of simple, measurable maps from some measurable space $S$ into $\Delta(Z)$, the \emph{independence} axiom takes the form:
\[
    f \succeq g \quad \iff \quad \alpha f + (1-\alpha)h \succeq \alpha g + (1-\alpha )h
\]
for $\alpha \in (0,1]$ and some act $h \in X$. Similarly, common weakenings of independence such as \emph{certainty independence} (\citealt{gilboa1989maxmin}), \emph{weak certainty independence} (\citealt{maccheroni2006ambiguity}), \emph{worst independence} (\citealt{chateauneuf2009ambiguity}), \emph{risk independence} (\citealt{cerreia2011uncertainty}) and so forth are all of this form.

\subsubsection{Stationarity}

Let $X = Z^\mathbb{N}$ denote the set of all infinite horizon consumption streams taking values in some set of prizes $Z$.  A preference on $X$ is said to be \emph{stationary} in the sense of \cite{koopmans1960} if:
\[
    (x_1, x_2, \ldots ) \succeq (x_1', x_2', \ldots ) \quad \iff \quad (z, x_1, x_2, \ldots) \succeq (z, x_1', x_2', \ldots )
\]
for all $z \in Z$. See also \cite{epstein1983stationary}.

\subsubsection{Convolution Invariance}

Suppose $X$ consists of all lotteries on $\mathbb{R}$ with bounded support.  \cite{mu2021monotone} consider continuous weak orders on $X$ that are monotone with respect to first-order stochastic dominance, and invariant under convolutions:
\[
    \mu \succeq \nu \quad \iff \quad \mu \ast \eta \succeq \nu \ast \eta,
\]
for all $\eta \in X$.\footnote{The term `additive' in the paper's title refers to this property when the preference is equivalently regarded as being defined over (bounded) random variables.} A preference on $X$  is said to exhibit \emph{constant absolute risk aversion} (e.g., \citealt{safra1998constant}) if:
\[
    \mu \succeq \nu \quad \iff \quad \mu \ast \delta_\alpha \succeq \nu \ast \delta_\alpha
\]
for all $\alpha \in \mathbb{R}$, where $\delta_\alpha$ denotes the Dirac measure centered at $\alpha$.\footnote{Similarly, \emph{constant relative risk aversion} is also a special case of $\mathcal{M}$-invariance, where $\mathcal{M}$ consists of the transformations which multiplicatively scale the support of a lottery.}

\subsubsection{Product \& Dilution Invariance}

Let $X$ consist of all finite Blackwell experiments on some fixed, finite set of states of the world $\Theta$. Thus an element of $X$ is a tuple $\big(S, \{\mu_\theta\}_{\theta \in \Theta}\big)$, where $S$ is a finite set of signals, and each $\mu_\theta$ is a probability measure on $S$.  \cite{pomatto2023cost} consider `costliness' orderings over $X$ that are invariant under two varieties of transformations: the formation of \emph{products}, and of so-called \emph{dilutions}. In this context, products of Blackwell experiments formalize the idea of running two simultaneous and independent experiments. The $\alpha$-dilution of an experiment, denoted $\alpha \cdot \big(S, \{\mu_\theta\}_{\theta \in \Theta}\big)$, is the experiment $\big(S \cup \{o\}, \{\mu'_\theta\}_{\theta \in \Theta}\big)$, where $o$ is a completely uninformative signal, and (i) $\mu'_\theta(A) = \alpha \mu_\theta(A)$ for all $A \subseteq S$, and $\mu_\theta'(\{o\}) = 1-\alpha$. This corresponds to invariance under:
\[
    \big(S, \{\mu_\theta\}_{\theta \in \Theta}\big) \succeq \big(S', \{\nu_\theta\}_{\theta \in \Theta}\big)  \iff  \big(S\times T, \{\mu_\theta \otimes \eta_\theta\}_{\theta \in \Theta}\big) \succeq \big(S'\times T, \{\nu_\theta \otimes \eta_\theta\}_{\theta \in \Theta}\big)
\]
for all $(T, \{\eta_\theta\}_{\theta \in \Theta}) \in X$, and:
\[
    \big(S, \{\mu_\theta\}_{\theta \in \Theta}\big) \succeq \big(S', \{\nu_\theta\}_{\theta \in \Theta}\big)  \iff \alpha \cdot \big(S, \{\mu_\theta\}_{\theta \in \Theta}\big) \succeq \alpha \cdot \big(S', \{\nu_\theta\}_{\theta \in \Theta}\big) 
\]
for all $\alpha \in (0,1]$.

\subsection{Knock-On Effects}

Consider an analyst who observes no revealed preference between two alternatives, $x$ and $y$.  As part of the process of attempting to construct an $\mathcal{M}$-invariant rationalizing preference $\succeq$ for the data, the analyst wishes to ascribe some relation between $x$ and $y$.\medskip

However, if the analyst wishes to add $x \succeq y$, since $\succeq$ is required to be $\mathcal{M}$-invariant, they are also compelled to add, in addition, every relation of the form $\omega(x) \succeq \omega(y)$, for each $\omega \in \mathcal{M}$.  We term these additional relations the {\bf knock-on effects} of adding $x \succeq y$. Even when adding $x \succeq y$ to the data leads to no inconsistency, these knock-on effects may themselves lead to the creation of cycles.

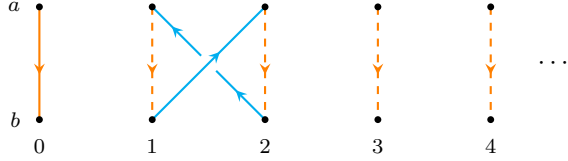
\begin{figure}[t]
    \centering
    \begin{tikzpicture}

\coordinate (A) at (0,0);
\coordinate (B) at (1.5,0);
\coordinate (C) at (3,0);
\coordinate (D) at (4.5,0);
\coordinate (E) at (6,0);
\coordinate (F) at (0,1.5);
\coordinate (G) at (1.5,1.5);
\coordinate (H) at (3,1.5);
\coordinate (I) at (4.5,1.5);
\coordinate (J) at (6,1.5);

{\footnotesize
\node[left] at (-.15,0) {$b$};
\node[right] at (6.15,0) {};
\node[left] at (-.15, 1.5) {$a$};
\node[below] at (0,-.15) {$0$};
\node[below] at (1.5,-.15) {$1$};
\node[below] at (3,-.15) {$2$};
\node[below] at (4.5,-.15) {$3$};
\node[below] at (6,-.15) {$4$};}
\node[right] at (6.5, .75) {$\mathbf{\cdots}$};

\begin{scope}[decoration={markings,mark = at position 0.6 with {\arrow{stealth}}}]
		\draw[line width =.8, cyan, postaction={decorate}] (C) -- (2.35, .65);
		\draw[line width =.8, cyan, postaction={decorate}] (2.15, .85) -- (G);
		\draw[line width =.8, cyan, postaction={decorate}] (B) -- (H) ;
		\draw[line width =.8, orange, postaction={decorate}] (F) -- (A) ;
		\draw[line width =.8, dashed, orange, postaction={decorate}] (G) -- (B) ;
		\draw[line width =.8, dashed, orange, postaction={decorate}] (H) -- (C) ;
		\draw[line width =.8, dashed, orange, postaction={decorate}] (I) -- (D) ;
		\draw[line width =.8, dashed, orange, postaction={decorate}] (J) -- (E) ;
 	 \end{scope}

\fill (A) circle [black!90, radius=1.3pt];
\fill (B) circle [black!90, radius=1.3pt];
\fill (C) circle [black!90, radius=1.3pt];
\fill (D) circle [black!90, radius=1.3pt];
\fill (E) circle [black!90, radius=1.3pt];
\fill (F) circle [black!90, radius=1.3pt];
\fill (G) circle [black!90, radius=1.3pt];
\fill (H) circle [black!90, radius=1.3pt];
\fill (I) circle [black!90, radius=1.3pt];
\fill (J) circle [black!90, radius=1.3pt];

\end{tikzpicture}
\caption{An illustration of \autoref{knockonexample}. The revealed preference is denoted in blue. A hypothetical comparison between $(b,0)$ and $(a,0)$ is shown (orange solid), and the resulting knock-on effects under $\mathcal{M}$ that would arise from adding this comparison (orange dashed). Note that while the comparison between $(b,0)$ and $(a,0)$ itself creates no cycles, its knock-on effects do.}
\end{figure}

\begin{example}\label{knockonexample}
    Consider a domain of dated rewards featuring two prizes, an apple $a$ and a banana $b$, that can be delivered to a consumer at any number of days in the future, i.e. $X = \{a,b\} \times \{0, 1,\ldots \}$.  We observe that a subject prefers:
    \[
        (a,1) \succ^R (b,2) \quad \textrm{ and } \quad (a,2) \succ^R (b,1).
    \]
    Let $\mathcal{M}$ denote all transformations of the form $(z,t) \mapsto (z,t+n)$ for each $n \ge 0$.  A preference is $\mathcal{M}$-invariant if and only if it is is \emph{stationary}, in the sense of \cite{fishburn1982}.\medskip

    Suppose we wish to extend this data set to incorporate a preference between the two prizes today, i.e. $(a,0)$ and $(b,0)$. Since the data itself is transitive, but makes no comparison between these alternatives, there exist rationalizing preference relations ranking both $(a,0) \succ (b,0)$ and $(b,0) \succ (a,0)$.\footnote{See, e.g., \citealt{dushnik1941}.}  However, every stationary (i.e. $\mathcal{M}$-invariant) rationalization must rank $(a,0) \succ (b,0)$.  To see this, note that if any stationary rationalization ranked $(b,0) \succeq (a,0)$, as knock-on effects it would necessarily also rank $(b,1) \succeq (a,1)$ and $(b,2) \succeq (a,2)$. Thus, in this extension, we would have:
    \[
        (a,1) \succ (b,2) \succeq (a,2) \succ (b,1) \succeq (a,1),
    \]
    implying it could not be a preference relation.  Thus every stationary rationalization must rank $(a,0) \succ (b,0)$. \hfill $\blacksquare$
\end{example}

\section{Characterizing Rationalizability: The Case of Commuting Transforms}\label{commutativesect}

In this section, we consider the special case in which each pair of transformations in $\mathcal{M}$ commute, i.e.:
\[
    (\omega \circ \omega')(x) = (\omega' \circ \omega)(x)
\]
for all $x \in X$ and $\omega,\omega' \in \mathcal{M}$.  Every example in Section 2.1.1, 2.1.2, and 2.1.5 is of this form, as are often families of transformations which depend only on a single real scalar, such as mixing under various weights with a fixed act or lottery (e.g.\ worst-independence, Section 2.1.3). In such cases, we refer to $\mathcal{M}$ as a {\bf commutative family}.\medskip

Let $R \subseteq X \times X$ be an arbitrary binary relation.  We define the $\mathcal{M}$-{\bf closure} of $R$, denoted $R_\mathcal{M}$ via:
\[
     \omega(x) \; R_{\mathcal{M}} \; \omega(y) \quad \iff \quad x \; R \; y,
\]
for some $\omega \in \mathcal{M}$. Since we have assumed that the identity function $\textrm{id} \in \mathcal{M}$, by setting $\omega = \textrm{id}$ we obtain that $R \subseteq R_\mathcal{M}$.\medskip

Consider now the data $\langle \succsim^R ,\succ^R\rangle$.  Our first main result says that, when $\mathcal{M}$ is a commutative family, the data $\langle \succsim^R, \succ^R\rangle$ are rationalizable by an $\mathcal{M}$-invariant preference relation if and only if its $\mathcal{M}$-closure is acyclic.

\begin{theorem}\label{commutativechar}
    Let $X$ be a set, and $\mathcal{M}$ an arbitrary family of commuting transformations. Then $\langle \succsim^R, \succ^R\rangle$ is rationalizable by an $\mathcal{M}$-invariant preference relation if and only if $\big \langle \succsim^R_\mathcal{M}, \succ^R_\mathcal{M}\big \rangle$ is acyclic.
\end{theorem}

Note that when $\mathcal{M} = \{\textrm{id}\}$, $\mathcal{M}$ is clearly a commutative family, and that \emph{every} preference is trivially $\mathcal{M}$-invariant.  Thus \autoref{commutativechar} strictly subsumes the classical characterization of \cite{richter1966revealed}. In \autoref{garpapp}, we explore connections between \autoref{commutativechar} and various well-known modifications of the generalized axiom of revealed preference in the special case of price-consumption data.\medskip

\subsection{Application: Probabilistic Sophistication}

Let $S$ denote a finite set of states of the world, and $X = 2^S$ the power set of $S$.  Elements of $X$ correspond to \emph{events}.  Consider a complete and transitive order $\succeq$ on $X$, which we interpret as an agent's subjective assessment of the relative likelihood of events (i.e.\ $A \succeq B$ denotes that the agent subjectively believes that $A$ is more likely than $B$). \medskip

Such an ordering is said to be a {\bf qualitative probability} if, for all events $A,B, C \in X$ with $C$ disjoint from $A\cup B$,
\[
    A \succeq B \quad \iff \quad A \cup C\;  \succeq \; B \cup C,
\]
and in addition, $A \subseteq B$ implies $B \succeq A$, and $S \succ \varnothing$.
We refer to a qualitative probability as {\bf probabilistically sophisticated} if it can be represented by a probability measure.\medskip

\cite{kraft1959intuitive} exhibit a qualitative probability over a five element state space which is not probabilistically sophisticated, disproving a conjecture of \cite{definetti51}. Using results on linear inequalities, they obtain an infinite system of `cancellation' conditions, which jointly characterize probabilistic sophistication. Despite the fact these conditions are not invariance axioms, by regarding $X$ as a subset of a richer domain, we may nonetheless use \autoref{commutativechar} to provide a simple test (cf.\ \citealt{epstein2000probabilities}).\medskip 

Let $\mathbb{Z}^S$ denote the set of all integer-valued functions on $S$, and $\mathcal{M}$ the commutative family consisting of the transformations $f \mapsto f + g$, for $g \in \mathbb{Z}^S$. By identifying elements of $X$ with their indicator functions, we may regard $X$ as a subset of $\mathbb{Z}^S$, and hence any order $\succeq$ on $X$ as an (incomplete) order $\succeq^*$ on $\mathbb{Z}^S$.  Any probability measure $\mu$ on $S$ defines a (i) complete, (ii) transitive, (iii) increasing, and (iv) $\mathcal{M}$-invariant order $\succeq$ on $\mathbb{Z}^S$ via:
\[
    f \succeq g \iff \int f \, d\mu \ge \int g \, d\mu.
\]
Conversely, however, not every order on $\mathbb{Z}^S$ satisfying (i) - (iv) can be represented by such a functional.\footnote{Such orders may fail to be \emph{Archimedean} in the sense of \cite{krantz1971foundations}, p.\ 73; for example, the lexicographic order on $\mathbb{Z}^2$ satisfies (i) - (iv), but has no representation of this form.} Nonetheless, every $\succeq^*$ whose $\mathcal{M}$-closure is acyclic can be extended to a preference on $\mathbb{Z}^S$ admitting such a representation.\footnote{This follows from Theorem 1.4 of \cite{scott1964measurement}.  Formally, Scott shows that a necessary and sufficient condition for $\succeq$ to be probabilistically sophisticated is for $\succeq^*$ to be able to be extended into a so-called ``strictly monotonic" order (\citealt{scott1964measurement}, p.\ 237). It is straightforward to show any $\mathcal{M}$-invariant preference on $\mathbb{Z}^S$ is strictly monotonic in Scott's sense.}

\begin{corollary}\label{probsophcorr}
    A qualitative probability $\succeq$ on $2^S$ is probabilistically sophisticated if and only if the $\mathcal{M}$-closure of $\succeq^*$ is acyclic.
\end{corollary}

In light of \autoref{probsophcorr}, Kraft et al.'s counterexample (\citealt{kraft1959intuitive}, p.\ 414) must feature some cycle in its $\mathcal{M}$-closure. Explicitly, their ordering includes the relations:  
\[
    \mathbbm{1}_{14} \prec^* \mathbbm{1}_{235}, \quad \mathbbm{1}_{23} \prec^* \mathbbm{1}_{15}, \quad \mathbbm{1}_{25} \prec^* \mathbbm{1}_{34}, \quad \textrm{ and } \quad \mathbbm{1}_{35} \prec^* \mathbbm{1}_2
\]
on $\mathbb{Z}^S$, where $S = \{1,\ldots, 5\}$. Thus in the $\mathcal{M}$ closure of $\succeq^*$, we obtain:
\[ \mathbbm{1}_{14} \prec^{*\mathcal{M}} \mathbbm{1}_{235} \prec^{*\mathcal{M}} \mathbbm{1}_{15} + \mathbbm{1}_5 \prec^{*\mathcal{M}} \mathbbm{1}_{1345}  - \mathbbm{1}_2 \prec^{*\mathcal{M}} \mathbbm{1}_{124} - \mathbbm{1}_{2} = \mathbbm{1}_{14},
\]
confirming that $\succeq$ is not probabilistically sophisticated.

\subsection{Preorder extensions}

One special case of our general problem (\hyperlink{q1}{Q.1}) is when an incomplete preference, or {\bf preorder} (i.e. a reflexive and transitive binary relation $\succeq$ on some set $X$), satisfying certain invariance axioms, can be extended into a preference relation with the same properties.  In this subsection, we illustrate how \autoref{commutativechar} may be used to obtain extension results of this form for preorders satisfying a number of natural economic properties.

\subsubsection{Additive Preorder Extensions}

Let $V$ denote a real vector space, and $C\subseteq V$ be a cone.  Suppose that $X \subseteq V$ is closed under addition by vectors in $C$, i.e. if $x\in X$ and $c \in C$, then $x+c\in X$.  Let $\mathcal{M}$ denote all transformations of the form $x \mapsto x + c$, for $c \in C$. In this circumstance, an $\mathcal{M}$-invariant preference is said to be $C$-{\bf additive}.\medskip

\begin{corollary}\label{additivecorr} Every $C$-additive preorder $\succeq$ admits a $C$-additive preference extension.\end{corollary}

\noindent \autoref{additivecorr} contains a number of well-known special cases, including:

\begin{enumerate}
\item[(i)] Quasilinearity over classical consumption spaces, i.e.\ $V=\mathbb{R}^{n+1}$, $C = \{(a,0,\ldots,0):a\geq 0\}$ and $X=\mathbb{R}^{n+1}_+$.  
\item[(ii)] Additive preferences over a vector space, i.e.\ $V$ is a real vector space and $V = X = C$.
\item[(iii)] CARA preferences over compactly supported monetary lotteries, i.e. $X = V = L^\infty$, and $C$ denotes the equivalence classes of almost-everywhere constant functions.  (Here, the preorder $\succeq$ is additionally presumed to be indifferent between any random variables which coincide in law.)
\end{enumerate}

\subsubsection{Homothetic preorder extensions}

Next, let $X$ denote a cone in a real vector space.  A preorder is \emph{homothetic} if, for every $x,y\in X$ and $\lambda > 0$, $x\succeq y$ implies $\lambda x \succeq \lambda y$, with a corresponding statement for strict preference.  It has been known since at least \citet{demuynck2009} that if $X$ is a cone in a Euclidean vector space, then every montonic and homothetic preorder has a monotonic and homothetic weak order extension.  The following corollary establishes a modest generalization of this result, removing both the assumption of monotonicity, and finite dimensionality:

\begin{corollary}\label{cor:nonmon}Every homothetic preorder has a homothetic preference extension.\end{corollary}

\noindent DeMuynck's result follows as any extension of a monotonic preorder is by definition monotonic as well.

\subsubsection{An Algebraic Version of \citet{dubra2004}}
   
Let $(Y,\Sigma)$ be some measurable space and let $\Delta(Y)$ be the set of countably additive probability measures on $(Y,\Sigma)$.  We say that a preorder $\succeq$ on $\Delta(Y)$ satisfies {\bf rational independence} if, for all $p,q,r\in\Delta(Y)$ and $\alpha \in \mathbb{Q}\cap (0,1]$, $p \succeq q$ if and only if $\alpha p + (1-\alpha)r \succeq \alpha q + (1-\alpha)r$.\medskip

Suppose that $\succeq$ is a rationally independent preorder. We define an extension $\succeq^*$ to the set of all signed measures of bounded variation as follows: let $\nu \succeq^* \nu'$ if and only if there exists some $\alpha\in\mathbb{Q}$, $\alpha > 0$, and $p,q\in\Delta(Y)$ with $p \succeq q$, such that $(\nu-\nu')=\alpha(p-q)$.\medskip

This extension is well-defined, as if $\nu-\nu'=\alpha(p-q)=\beta(r-s)$ and $p\succeq q$, by rational independence and transitivity, it cannot be $s \succ r$.  Thus, in particular, $\succeq^*$ may be regarded as extending $\succeq$.\footnote{To be formal, $\succeq^*$ extends the image of $\succeq$ under the inclusion map taking $\Delta(Y)$ into the space of signed measures.}  Moreover, if $\nu_1\succeq^* \nu_2 \succeq^* \nu_3$, it is straightforward to establish that $\nu_1 \succeq^* \nu_3$. Thus $\succeq^*$ is itself transitive.\medskip

Let $\mathcal{M}$ denote the maps $\nu \mapsto \nu + \rho$, for each signed measure $\rho$. By construction, we have $\nu \succeq^* \nu'$ if and only if $\nu + \rho \succeq^* \nu'+\rho$, for any $\rho$.  Thus $\succeq^*$ is its own $\mathcal{M}$-closure. Since $\mathcal{M}$ is a commutative family, \autoref{commutativechar} provides the existence of an $\mathcal{M}$-invariant extension, $\succeq^{**}$, of $\succeq^*$.\medskip

Finally, note the restriction of $\succeq^{**}$ to $\Delta(Y)$ itself satisfies rational independence:  if $p \succeq^{**} q$ and $\alpha \in \mathbb{Q}\cap (0,1]$, then we must have $\alpha p \succeq^{**} \alpha q$.\footnote{This follows from a straightforward induction argument.}   Consequently,  we obtain that $\alpha p+(1-\alpha)r \succeq' \alpha q + (1-\alpha)r$, as desired, with an analogous statement holding for strict preference.  

\begin{corollary}Every rationally independent preorder has a rationally independent preference extension.\end{corollary}

\section{The General Case}\label{noncommcase}

\subsection{Overview}

In the preceding section, \autoref{commutativechar} showed that when the transformations in $\mathcal{M}$ commute with each other, a simple generalization of \cite{richter1966revealed}'s acyclicity condition characterizes rationalizability by an $\mathcal{M}$-invariant preference.  However, when $\mathcal{M}$ is not a commutative family, this conclusion fails. We return to the example from the introduction.

\begin{example}\label{koopmansrevisited}
    Let $a,b,c,d \in Z$ be prizes, and suppose $X$ consists of all infinite horizon consumption streams taking values in $Z$, i.e.\ $X = Z^\mathbb{N}$.  Let $\mathcal{M}$ consist of all finite compositions of the transformations $\{\omega_z\}_{z \in Z}$ which append the prize $z$ to the start of a consumption stream. Here, $\mathcal{M}$-invariance corresponds to the stationarity of a preference in the sense of \cite{koopmans1960}.\medskip
    
    Let $x,y \in X$ be arbitrary consumption streams, and recall in \autoref{stationaritycex} we observed:
    \begin{equation}\tag{1}
        \begin{aligned}
            (a,x_1,\ldots) & \succ^R (b, y_1, \ldots)\\
            (b,x_1,\ldots) & \succ^R (a, y_1, \ldots),
        \end{aligned}
    \end{equation}
    and
    \begin{equation}\tag{2}
        \begin{aligned}
            (c,y_1,\ldots) & \succ^R (d, x_1, \ldots)\\
            (d,y_1,\ldots) & \succ^R (c, x_1, \ldots).
        \end{aligned}
    \end{equation}
This relation is vacuously transitive, as is its $\mathcal{M}$-closure. However, as noted in \autoref{stationaritycex}, it cannot be rationalized by any stationary preference. Notably, this is \emph{purely} a consequence of the failure of these transformations to commute.\medskip

To see this, consider an abstract domain $X$ with alternatives $x$ and $x'$, and let $\mathcal{M}$ be generated by all finite compositions of four abstract transformations, $\omega_a, \omega_{b}, \omega_c, \omega_{d}$, each mapping $X \to X$. In this reformulation, our data relation is again given by:
\[
    \omega_{b}(x) \succ^R \omega_a (y) \quad \omega_{a} (x) \succ^R \omega_{b} (y)
\]
and
\[
    \omega_{d}(y) \succ^R \omega_{c} (x) \quad \omega_{c} (y) \succ^R \omega_{d} (x),
\]
but crucially, suppose now that these four transformations were to \emph{commute}. It straightforwardly follows that, in the $\mathcal{M}$-closure, we now have a cycle:
\[
    \omega_{ad}(x) \succ^R_\mathcal{M} \omega_{bd}(y) \succ^R_\mathcal{M}  \omega_{bc}(x) \succ^R_\mathcal{M} \omega_{ac}(y) \succ^R_\mathcal{M} \omega_{ad}(x),
\]
where $\omega_{ij} \equiv \omega_i \circ \omega_j$. Thus the inability of the $\mathcal{M}$-closure to detect obstructions to rationalizability arise solely from the failure of the transformations in $\mathcal{M}$ to commute. \hfill $\blacksquare$
\end{example}

The essence of \autoref{koopmansrevisited} is that, when $\mathcal{M}$ does not consist of commuting transformations, we may be able to falsify rationalizability only on the basis of obtaining mutually unsatisfiable \emph{out-of-sample} restrictions on the comparisons any extension must make.\medskip

In this section, we introduce a strengthening of the transitive closure.  In the classical setting (i.e.\ when $\mathcal{M} = \{\textrm{id}\}$) the transitive closure encodes all the restrictions on out-of-sample comparisons that any rationalizing preference must make.  In contrast, the potential infinitude of knock-on effects require our notion of closure to operate not on relations between single pairs of alternatives, but rather on \emph{sets} of simultaneous constraints, jointly imposed by the data and the structure of $\mathcal{M}$. We show that this generalized notion of closure characterizes both (i) the existence of $\mathcal{M}$-invariant rationalizations, for \emph{any} choices of $X$, $\mathcal{M}$, and $\langle \succsim^R, \succ^R\rangle$, as well as (ii) the out-of-sample predictions generated by $\mathcal{M}$-invariance, just as the transitive closure does in the classical setting.

\subsection{Broken Cycles \& Forbidden Subrelations}

For some $N \ge 1$, let $\omega_1,\ldots, \omega_N \in \mathcal{M}$, and $x_1, y_1, \ldots, x_N, y_N \in X$ be a sequence of $\succsim^R$-unrelated pairs (i.e.\ $x_i$ and $y_i$ are $\succsim^R$-unrelated). We term a collection of relations:
\begin{equation}\label{bc}
 \begin{aligned}
        \omega_1(x_1) & \succsim^{R}_\intercal \omega_2(y_2)\\
        \omega_2(x_2) & \succsim^{R}_\intercal \omega_3(y_3)\\
        & \; \; \, \vdots \\
        \omega_{N-1}(x_{N-1}) & \succsim^{R}_\intercal \omega_N(y_N)\\
        \omega_N(x_N) & \succsim^{R}_\intercal \omega_1(y_1),
    \end{aligned} 
\end{equation}
a {\bf broken cycle.}  If any of the relations $\succsim^{R}_\intercal$ also belongs to $\succ^{R}_\intercal$, we say \eqref{bc} forms a {\bf strict broken cycle}.  Any broken cycle (strict or otherwise) implies joint restrictions on the possible comparisons any $\mathcal{M}$-invariant rationalization may make between the $x_i$ and $y_i$. To formalize this, we say an order pair $\langle W,S \rangle$ defines a {\bf forbidden subrelation} for the broken cycle \eqref{bc} if:
\begin{itemize}
\item[(i)] The relation $W$ equals the set of all distinct pairs $(y_i, x_i)$.\footnote{The definition of a broken cycles allows that, for some $i \neq j$, it may be the case that $(y_i, x_i) = (y_j, x_j)$.}
\item[(ii)] If the broken cycle is not strict, then $\varnothing \subsetneq S \subseteq W$.\footnote{Recall that by definition of an order pair, $S\subseteq W$, regardless of whether \eqref{bc} is strict.}
\end{itemize}

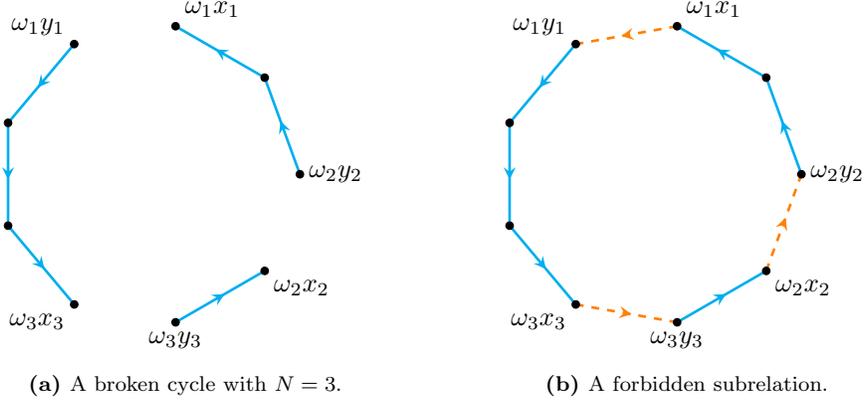
\begin{figure}[t]
    \centering
    \begin{subfigure}[c]{0.45\textwidth}
    \centering
    \begin{tikzpicture}

 \begin{scope}[decoration={markings,mark = at position 0.55 with {\arrow{stealth}}}]
		
		\draw[line width =1, cyan, postaction={decorate}] (0: 2cm) -- (360/9: 2cm);
		\draw[line width =1, cyan, postaction={decorate}] (360/9: 2cm) -- (2*360/9: 2cm);
		
		\draw[line width =1, cyan, postaction={decorate}] (3*360/9: 2cm) -- (4*360/9: 2cm);
		\draw[line width =1, cyan, postaction={decorate}] (4*360/9: 2cm) -- (5*360/9: 2cm);
		\draw[line width =1, cyan, postaction={decorate}] (5*360/9: 2cm) -- (6*360/9: 2cm);
		
		\draw[line width =1, cyan, postaction={decorate}] (7*360/9: 2cm) -- (8*360/9: 2cm);

		
 	 \end{scope}

	\draw[fill] (0: 2cm) circle [radius=0.05];
	\draw[fill] (360/9: 2cm) circle [radius=0.05];
	\draw[fill] (2*360/9: 2cm) circle [radius=0.05];
	\draw[fill] (3*360/9: 2cm) circle [radius=0.05];
	\draw[fill] (4*360/9: 2cm) circle [radius=0.05];
	\draw[fill] (5*360/9: 2cm) circle [radius=0.05];
	\draw[fill] (6*360/9: 2cm) circle [radius=0.05];
	\draw[fill] (7*360/9: 2cm) circle [radius=0.05];
	\draw[fill] (8*360/9: 2cm) circle [radius=0.05];

	\node [right] at (0:2cm) {$\omega_2 y_2$};
	\node [above right] at (2*360/9:2cm) {$\omega_1 x_1$};
	\node [above left] at (3*360/9:2cm) {$\omega_1 y_1$};
	\node [above] at (4*360/9:2cm) {};
	\node [above] at (5*360/9:2cm) {};
	\node [below left] at (6*360/9:2cm) {$\omega_3 x_3$};
	\node [below ] at (7*360/9:2cm) {$\omega_3 y_3$};
	\node [below right] at (8*360/9:2cm) {$\omega_2 x_2$};

\end{tikzpicture}
\subcaption{A broken cycle with $N=3$.}
\end{subfigure}\hfill
\begin{subfigure}[c]{.45\textwidth}
    \centering
    \begin{tikzpicture}

 \begin{scope}[decoration={markings,mark = at position 0.55 with {\arrow{stealth}}}]
		
		\draw[line width =1, cyan, postaction={decorate}] (0: 2cm) -- (360/9: 2cm);
		\draw[line width =1, cyan, postaction={decorate}] (360/9: 2cm) -- (2*360/9: 2cm);
		\draw[line width =1, orange, dashed, postaction={decorate}] (2*360/9: 2cm) -- (3*360/9: 2cm);
		\draw[line width =1, cyan, postaction={decorate}] (3*360/9: 2cm) -- (4*360/9: 2cm);
		\draw[line width =1, cyan, postaction={decorate}] (4*360/9: 2cm) -- (5*360/9: 2cm);
		\draw[line width =1, cyan, postaction={decorate}] (5*360/9: 2cm) -- (6*360/9: 2cm);
		\draw[line width =1, orange, dashed, postaction={decorate}] (6*360/9: 2cm) -- (7*360/9: 2cm);
		\draw[line width =1, cyan, postaction={decorate}] (7*360/9: 2cm) -- (8*360/9: 2cm);
		\draw[line width =1, orange, dashed, postaction={decorate}] (8*360/9: 2cm) -- (0: 2cm);

		
 	 \end{scope}

	\draw[fill] (0: 2cm) circle [radius=0.05];
	\draw[fill] (360/9: 2cm) circle [radius=0.05];
	\draw[fill] (2*360/9: 2cm) circle [radius=0.05];
	\draw[fill] (3*360/9: 2cm) circle [radius=0.05];
	\draw[fill] (4*360/9: 2cm) circle [radius=0.05];
	\draw[fill] (5*360/9: 2cm) circle [radius=0.05];
	\draw[fill] (6*360/9: 2cm) circle [radius=0.05];
	\draw[fill] (7*360/9: 2cm) circle [radius=0.05];
	\draw[fill] (8*360/9: 2cm) circle [radius=0.05];

	\node [right] at (0:2cm) {$\omega_2 y_2$};
	\node [above right] at (2*360/9:2cm) {$\omega_1 x_1$};
	\node [above left] at (3*360/9:2cm) {$\omega_1 y_1$};
	\node [above] at (4*360/9:2cm) {};
	\node [above] at (5*360/9:2cm) {};
	\node [below left] at (6*360/9:2cm) {$\omega_3 x_3$};
	\node [below ] at (7*360/9:2cm) {$\omega_3 y_3$};
	\node [below right] at (8*360/9:2cm) {$\omega_2 x_2$};

\end{tikzpicture}
\subcaption{A forbidden subrelation.}
\end{subfigure}
\caption{A broken cycle (blue) and an associated forbidden subrelation for the broken cycle (dashed orange). Here, the illustrated forbidden subrelation is given by $W = S = \big\{(y_1,x_1), \ldots, (y_3,x_3)\big\}$.}
\end{figure}

The first relation, $W$, reflects restrictions on the possible \emph{weak} comparisons any extending preference can make. Similarly, $S$ encodes restrictions pertaining to \emph{strict} comparisons. Together, these relations capture \emph{set-valued} restrictions on the extension problem: if a binary relation $\succeq$ were to extend some forbidden subrelation $\langle W,S\rangle$, it would imply that (i) $W \subseteq \; \succeq$, and (ii) $S \subseteq \; \succ$, and hence that $\succeq$ necessarily contains a cycle.  Informally, a forbidden subrelation $\langle W,S\rangle$ may be read as the restriction:
\[
    \textrm{``Cannot simultaneously have } \underbrace{\cdots, y_{i_j} \succeq x_{i_j}, \cdots}_{\textrm{Relations in } W \setminus S}  \textrm{ and }  \underbrace{\cdots, y_{i_k} \succ x_{i_k}, \cdots}_{\textrm{Relations in } S}.\textrm{"}
\]
\noindent As the following example shows, these order pairs capture richer systems of restrictions than does the transitive closure.\medskip

\begin{example}
    Suppose first that we observe only $x \succsim^R y$ and $y \succsim^R z$.  Thus $x \succsim^R_\intercal z$, which itself defines a (non-strict) broken cycle, where $N=1$.  From this, we obtain $W = \big\{(z,x)\big\}$ and $S = \big\{(z,x)\big\}$ as a forbidden subrelation, encoding that no extension of the data can rank $z \succ x$.\footnote{More generally, for any $x_1,\ldots, x_N \in X$, since $\succsim^R$ is assumed reflexive, the collection of relations:
\[
	\begin{aligned}
		x_1 \succsim^R x_1 \\
		\vdots \quad \quad \\
		x_N \succsim x_N
	\end{aligned}
\]
defines a broken cycle, whose forbidden subrelations are those cyclic binary relations $x_1 \succeq \cdots \succeq x_N \succeq x_1$, with at least one strict component, i.e.\ the \emph{cycles} over $x_1$ to $x_N$.}\medskip

    If instead we observed that $\omega(x) \succsim^R y$ and $y \succsim^R \omega(z)$, by analogous reasoning we would obtain $W = \big\{(z,x)\big\}$ and $S = \big\{(z,x)\big\}$ again as a forbidden subrelation, as well as $W' = \big\{(\omega(z),\omega(x))\big\}$ and $S' = \big\{(\omega(z),\omega(x))\big\}$ reflecting both the direct implication of transitivity, that no extension can rank $\omega(z) \succ \omega(x)$, as well as the knock-on effects of this observation.\medskip

    More generally, however, it can be the case that multiple knock-on effects arising from even a single added comparisons can `complete' a broken cycle. The value of forbidden subrelations is that these capture the restrictions arising from these knock-on effects.  In \autoref{koopmansrevisited}, we observed, e.g., that $\omega_{x;}(\sigma) \succ^R \omega_x(\sigma')$ and $\omega_x(\sigma) \succ^R \omega_{x'}(\sigma')$.  These comparisons form a (strict) broken cycle, where $x_1 = x_2 = \sigma$ and $y_1 = y_2 = \sigma'$.  From this, we obtain as a forbidden subrelation $W = \big\{(\sigma', \sigma)\big\}$ and $S = \varnothing$. The fact $W$ and $S$ have cardinality less than the number of `gaps' $N$ of the broken cycle reflects the fact that multiple knock-on effects, arising from adding the relation in $W$, would jointly complete the broken cycle.\hfill $\blacksquare$
\end{example}

\subsection{The Collapse Operation}

Given the data $\langle \succsim^R, \succ^R\rangle$, let $\mathcal{F}$ denote the set of all forbidden subrelations obtained from broken cycles.  While pairs in $\mathcal{F}$ reflect the rich, set-valued constraints imposed by the algebraic structure of $\mathcal{M}$ on the extension problem, there are, nonetheless, further constraints that need not arise directly from broken cycles in this fashion.\medskip  

\begin{example}\label{collapseexample}
Suppose, based on two broken cycles, we obtain restrictions:
\[
    W_1 = \big\{(x,y), (x',y')\big\} \quad S_1 = \big\{(x',y')\big\}
\]
and
\[
    W_2 = \big\{(y,x), (x'', y'')\big\} \quad S_2 = \big\{(y,x)\big\},
\]
i.e.\ restrictions saying that we cannot simultaneously have $x \succeq y$ and $x' \succ y'$ in any extension (corresponding to $\langle W_1, S_1\rangle$) and similarly, that we cannot have $y \succ x$ and $x'' \succeq y''$ simultaneously (from $\langle W_2, S_2\rangle$).\medskip

Since every rationalizing preference relation must either rank $x \succeq y$ or $y \succ x$, these two pairs imply an additional, \emph{indirect} restriction on the extension problem: no extension can simultaneously rank both $x' \succ y'$ and $x'' \succeq y''$. Should any complete, extension contain both these comparisons, it must necessarily also rank either $x \succeq y$ or $y\succ x$ and hence extend either $\langle W_1, S_1\rangle$ or $\langle W_2,S_2\rangle$, and therefore cannot be transitive.  This indirect restriction can be expressed as the order pair:
\[
    \tilde{W} = \big\{(x', y'), (x'', y'')\big\} \quad \tilde{S} = \big\{(x', y')\big\},
\]
formed by deleting all instances of the mutually exhaustive pair, and taking the union of the respective remaining relations.\medskip

In fact, nothing material in this argument would be affected if, instead of $\langle W_1, S_1\rangle$ and $\langle W_2, S_2 \rangle$, we instead observed:
\[
    W_1' = \big\{(\omega(x),\omega(y)), (x',y')\big\} \quad S_1 = \big\{(x',y')\big\}
\]
and
\[
    W_2 = \big\{(\omega'(y),\omega'(x)), (x'', y'')\big\} \quad S_2 = \big\{(\omega'(y),\omega'(x))\big\},
\]
corresponding to the restrictions ``cannot have $\omega(x) \succeq \omega(y)$ and $x' \succ y'$ simultaneously," and ``cannot have $\omega'(y) \succ \omega'(x)$ and $x'' \succeq y''$ simultaneously." Any $\mathcal{M}$-invariant rationalization must still rank either $x \succeq y$ or $y \succ x$ and hence the same indirect restriction $\langle \tilde{W}, \tilde{S}\rangle$ obtains.\hfill $\blacksquare$
\end{example}

To formalize the observation underpinning \autoref{collapseexample}, we introduce a partial binary operation on finite order pairs.  Formally, we say a finite order pair $\langle \tilde{W}, \tilde{S}\rangle$ is a {\bf collapse} of two pairs $\langle W_1, S_1\rangle$ and $\langle W_2, S_2\rangle$ if:
\begin{itemize}
\item[(i)] For some $\omega, \omega' \in \mathcal{M}$ and $x,y \in X$,
\[
	\big(\omega(x), \omega(y)\big) \in W_i \setminus S_i \quad \textrm{ and } \quad \big(\omega'(y), \omega'(x)\big) \in W_j,
\]
where $i,j \in \{1,2\}$ are distinct.\medskip
\item[(ii)] The relations $\tilde{W}$ and $\tilde{S}$ are given by:
\[
\begin{aligned}
	\tilde{W} = \big[W_i \setminus \big(\omega(x), \omega(y)\big)\big] \cup \big[W_j \setminus \big(\omega'(y), \omega'(x)\big)\big]\\
\end{aligned}
\]
and
\[
\tilde{S} = S_i \cup \big[S_j \setminus \big(\omega'(y), \omega'(x)\big)\big],
\]
where we intentionally omit the curly braces on singleton sets of pairs to conserve notation.
\end{itemize}
Condition (i) requires that, modulo the transformations $\omega$ and $\omega'$, the two order pairs contain a mutually exhaustive pair of rankings of two alternatives, $x$ and $y$.  This requires that pair $\langle W_j, S_j\rangle$ contain a restriction of the form $\omega'(y) \succeq \omega'(x)$, for some $\omega'$, but allows $\langle W_i, S_i\rangle$ to contain a restriction of the form $\omega(x) \succeq \omega(y)$ or $\omega(x) \succ \omega(y)$.\footnote{This asymmetry in the treatment of $W_i$ and $W_j$ ensures that at least one of restriction on $x$ and $y$ is non-strict. If, e.g., $W_i$ said ``cannot have $x \succ y$ and $W_j$ said ``cannot have $y \succ x$, these restrictions would not be mutually exhaustive.}\medskip

Condition (ii) then defines each relation in the collapse as the union of the respective `parent' relations, minus the mutually exhaustive pair. Generally, an incomplete, $\mathcal{M}$-invariant extension of $\langle \tilde{W}, \tilde{S}\rangle$ need not extend either $\langle W_1, S_1\rangle$ or $\langle W_2, S_2\rangle$. However, such an extension can never be completed to an $\mathcal{M}$-invariant preference, as such a preference rank $x$ and $y$ eventually, at which point it must necessarily extend one of the initial forbidden subrelations, creating a cycle.\medskip

By repeatedly collapsing restriction order pairs, we are able to uncover further `indirect' restrictions on the extension problem.  To formalize this idea, let $\mathcal{F}^0  =\mathcal{F}$ denote the set of all forbidden subrelations arising from some broken cycle in the data, and inductively define:
\[
	\mathcal{F}^n = \big\{\langle W,S \rangle : \langle W,S\rangle \textrm{ is collapse of pairs in } \mathcal{F}^{n-1} \big\} \cup \mathcal{F}^{n-1}.
\]
Let:
\[
\mathcal{F}^* = \bigcup_{n\ge 0} \mathcal{F}^n.
\]
We say that the data $\langle \succsim^R, \succ^R\rangle$ are {\bf strongly acyclic} if the empty order pair $\langle \varnothing, \varnothing\rangle \not \in \mathcal{F}^*$.\medskip

\begin{example}
    Consider again the revealed preference in \autoref{koopmansrevisited}. Equations \eqref{stationaritycex1} and \eqref{stationaritycex2} define broken cycles, which respectively yield forbidden subrelations $W_1 = \big\{(\sigma' , \sigma)\big\}$ and $W_2 = \big\{(\sigma, \sigma')\big\}$, with $S_1=S_2 = \varnothing$.  The collapse of these pairs is precisely $\tilde{W} = \tilde{S} = \varnothing$, and hence the data fail to be strongly acyclic. This reflects our earlier observation that even though the data are transitive (and have an acyclic $\mathcal{M}$-closure), any stationary rationalization must compare $\sigma$ and $\sigma'$, and hence complete one of the broken cycles via the knock-on effects arising from this comparison. \hfill $\blacksquare$
\end{example}

Strong acyclicity is a necessary condition for $\mathcal{M}$-invariant rationalizability: if the empty pair belongs to $\mathcal{F}^*$, then it belongs to some $\mathcal{F}^n$, and hence through some finite sequence of collapses, can be derived from broken cycles in the data. Since \emph{every} extension of the data must necessarily extend the empty pair, this means that no possible extension of the data can ever be completed without creating a cycle at some stage along the way.\medskip

A priori, however, it is unclear whether $\mathcal{F}^*$, the set of restrictions generated by repeated application of the collapse, reflects all the restrictions on the extension problem. The following theorem is the primary result of this paper.  It says that $\mathcal{M}$-acyclicity is not only necessary, but in fact fully characterizes the existence of an $\mathcal{M}$-invariant rationalizing preference, without \emph{any} assumptions on the structure of the domain $X$, transformations $\mathcal{M}$, or data $\langle \succsim^R, \succ^R\rangle$.\medskip

\begin{theorem}\label{generalchar}
    The data $\langle \succsim^R, \succ^R\rangle$ are strongly acyclic if and only if they are rationalizable by an $\mathcal{M}$-invariant preference. 
\end{theorem}

\section{Invariance Under Partial Transformations}

Thus far, we have defined an an invariance with respect to a family of transformations, each mapping $X \to X$. However, a number of natural, economic axioms are of a slightly weakened form, requiring invariance under a given transformation $\omega$ to hold only for pairs of alternatives belonging to a subset $D_\omega \subseteq X$ which, crucially, does not depend on the preference itself.\footnote{This distinguishes them from, e.g., convexity or betweenness-type axioms.}\medskip

Formally, let $\varnothing \subseteq D_\omega \subseteq X$ and $\omega: D_\omega \to X$. We say a preference relation $\succeq$ is {\bf invariant} under $\omega$ if, for all $x,y \in D_\omega$:
\[
    x \succeq y \quad \iff \quad \omega(x) \succeq \omega(y).
\]
We term such an $\omega$ a {\bf partial transformation} of $X$, and refer to $D_\omega$ as the {\bf domain} of $\omega$.  Given two partial transformations $\omega$ and $\omega'$, their composition is the partial transformation $\omega' \circ \omega: D_{\omega' \circ \omega} \to X$, where $D_{\omega' \circ \omega}$ consists of those alternatives in $D_\omega$ whose image under $\omega$ belongs to $D_{\omega'}$.\footnote{Note that $D_{\omega' \circ \omega}$ may be empty. In this case, any preference is trivially invariant with respect to $\omega' \circ \omega$.}  Finally, given a collection $\mathcal{M}$ of partial transformations, we say that $\succeq$ is $\mathcal{M}$-invariant if it is invariant under every partial transformation in $\mathcal{M}$.  Without loss of generality, we will again assume that any such collection $\mathcal{M}$ (i) contains the identity transformation $X \to X$, and (ii) is closed under composition.

\subsection{Examples of Partial Invariance Axioms}

Despite their technical nature, it is perhaps surprising that a number of classical, economically important axioms may be regarded as invariances with respect to families of partial transforms.

\subsubsection{Ordinal Additivity}\label{ordadd}

Let $S$ denote a set of states of the world, and let $X = 2^S$.  For each $A \in X$, let:
\[
    D_A = \big\{ B \in X : A \cap B = \varnothing \big\},
\]
and define $\omega_A : D_A \to X$ via $\omega_A(B) = A \cup B$.  Recall a complete and transitive binary relation $\succeq$ on $X$, representing a subjective `likelihood ordering' over the events $X$, defines a \emph{qualitative probability} if it is invariant under every partial transformation $\omega_A$ (e.g.\ \citealt{definetti51, kraft1959intuitive}), and is additionally monotone with respect to set inclusion and nontrivial.

\subsubsection{Additive Separability}\label{addsep}

Let $X = \times_{i \in I} X_i$.  For any subset $A \subseteq I$, and any $x,y \in X$, we define $x_Ay \in X$ via:
\[
    (x_A y)_i = \begin{cases}x_i & \textrm{ if } i \in A \\ y_i & \textrm{ if } i \not \in A \end{cases}.
\]
A preference relation $\succeq$ on $X$ is said to be {\bf separable} (\citealt{leontief1947introduction, debreu1959topological, luce1964simultaneous}) if, for all $A \subseteq I$, and all $x,y,z,z' \in X$,
\[
    x_Az \succeq y_A z \quad \iff \quad x_Az' \succeq y_Az'.
\]
Separability may equivalently be expressed as an invariance axiom. First, let:
\[
    D_A^z = \{x \in X: x = x_A z\},
\]
denote those elements of $X$ equal to $z$ on $A$.  For each $A \subseteq I$ and $z,z' \in X$, define $\omega_A^{z \to z'}: D_A^z \to X$ via:
\[
    \omega_A^{z\to z'}(x_Az) = x_Az'.
\]
In other words, $\omega_A^{z\to z'}$ takes in an element of $X$ equal to $z$ on $A$, and replaces its values on $A$ instead with the restriction of $z'$.  Letting $\mathcal{M}$ denote these transforms (and their compositions), we obtain that a preference is $\mathcal{M}$-invariant if and only if it is separable.

\subsubsection{Savage's P2}

Closely related to \autoref{addsep}, let $S$ denote a set of states of the world, $\mathcal{X}$ a set of consequences, and $X = \mathcal{X}^S$ the set of all acts mapping $S \to \mathcal{X}$.  For any three acts $f,g$, and $h$ in $X$ and $A \subseteq S$, we analogously define:
\[
    f_Ah = \begin{cases}
        f(s) & \textrm{ if } s \in A\\
        h(s) & \textrm{ if } s \not \in A.
    \end{cases}
\]
A preference relation $\succeq$ on $X$ satisfies Axiom P2 (\citealt{savage1954foundations}) if and only if, for all $f,g,h,h' \in X$ and $A \subseteq S$, we have:
\[
    f_Ah \succeq g_Ah \quad \iff \quad f_Ah' \succeq g_Ah'.
\]
By precisely the same construction as in \autoref{addsep} we may rephrase this condition as an invariance axiom for a suitable family of partial transformations.

\subsubsection{Comonotonic Independence}

Let $\{1,\ldots, S\}$ denote a finite set of states of the world and $X = \mathbb{R}^S$ the set of all monetary acts.  Two acts $f,g \in X$ are said to be \emph{comonotonic} if it is never the case that:
\[
    f(s) > f(s') \quad \textrm{ and } \quad g(s) < g(s')
\]
for any $s,s' \in S$.  A preference relation satisfies {\bf comonotonic independence} (\citealt{schmeidler1989subjective}) if, for any pairwise comonotonic acts $f,g,h \in X$ and $\alpha \in (0, 1]$ we have:
\[
    f \succeq g \quad \iff \quad \alpha f + (1-\alpha) h \succeq \alpha g + (1-\alpha)h.
\]
Let $\mathfrak{S}$ denote the set of all permutations of $\{1,\ldots,  S \}$.  For $\sigma \in \mathfrak{S}$, let:
\[
    D_\sigma = \big\{f \in X : f(i) \ge f(j) \iff \sigma(i) > \sigma(j)\big \}.
\]
For each $\sigma \in \mathfrak{S}$, $\alpha \in (0,1]$ and $h \in D_\sigma$, define the transformation:
\[
    \omega_\sigma^{\alpha, h}: D_\sigma \to X
\]
via $f \mapsto \alpha f + (1-\alpha)h$.  Letting $\mathcal{M}$ denote the collection of all such transformations and their compositions, we obtain that a preference is $\mathcal{M}$-invariant if and only if it is comonotonic independent.



\subsection{Strong Acyclicity and Partial Transforms}

The primary result of this section is that \emph{precisely} the same notion of strong acyclicity, as defined in \autoref{noncommcase}, characterizes $\mathcal{M}$-invariant rationalizability for arbitrary families of partial transformations as well.  Formally, this simply follows from observing that no definition or proof step required each transformation $\omega \in \mathcal{M}$ to have $D_\omega = X$. However, we will illustrate this lack of dependence on domain in more detail.\medskip

Firstly, consider a broken cycle such as in \eqref{bc}. By definition, broken cycles require only that each $\omega_i(x_i)$ and $\omega_i(y_i)$ are defined (i.e.\ $x_i,y_i \in D_{\omega_i}$).  Given any well-defined broken cycle, obtaining its forbidden subrelations does not require us to apply any transform, partial or otherwise. Instead, rather, we take only pre-images: if $\omega_i(y_i)$ and $\omega_i(x_i)$ are gaps in a broken cycle, then our forbidden subrelation contains the pair $(y_i, x_i)$. This passage is always well-defined so long as $\omega_i(y_i)$ and $\omega_i(x_i)$ are. Therefore, even when $\mathcal{M}$ contains partial transforms, we may obtain forbidden subrelations from broken cycles in a precisely the same manner.\medskip

Suppose now that we have a two finite order pairs $\langle W_1, S_1\rangle$ and $\langle W_2,S_2 \rangle$. As with the definitions of forbidden subrelations, the collapse operation never requires us to apply a common transform to two alternatives, but rather allows us to regard two pairs as `mutually exhaustive' so long as they are both images of (up to) two different transforms. Once again, this is unaffected by incompleteness of the domain of these transformations.  Thus the collapse operation, and with it the `collapsed closure' $\mathcal{F}^*$ also remain well-defined, even when $\mathcal{M}$ contains partial transformations.\medskip

Despite these observations, it is unclear however, a priori, whether letting $\mathcal{M}$ consist of partial transformations can potentially lead to novel obstructions to rationalizability. The following theorem, whose proof follows immediately from the observation that the proof of \autoref{generalchar} remains valid for the case of partial transforms, establishes that this is not the case.\medskip

\begin{theorem}\label{gencharpartialfn}
    Let $\mathcal{M}$ consist of partial transformations. Then the data $\langle \succsim^R, \succ^R\rangle$ are strongly acyclic if and only if they are rationalizable by an $\mathcal{M}$-invariant preference. 
\end{theorem}

\noindent We illustrate this result below, in the context of the one-urn paradox of \cite{ellsberg1961risk}.

\begin{example}[One-Urn Ellsberg Paradox]
    There is an urn containing $90$ balls. $30$ of these balls are red, and the remainder are either yellow, or black. A ball is to be (uniformly) randomly drawn from the urn; let $S = \{s_r, s_y , s_b\}$ denote the color of this ball, and let $X = 2^S$.\medskip

    Suppose we elicit that a subject believes:
    \[
        s_r \; \succ^R s_b
    \]
and
    \[
        s_r \cup s_y \; \succ^R  s_r \cup s_y
    \]
    where to conserve on notation we have omitted curly braces from singleton sets, and $\succeq^R$ denotes `subjectively deemed more likely.' These observations are trivially inconsistent with any qualitative probability on $X$. To see this via \autoref{gencharpartialfn} observe that, in the notation of \autoref{ordadd}, we have:
    \[
        s_r \; \succ^R s_b
    \]
    and
    \[
        \omega_{s_y}(s_b) \succ^R \omega_{s_y}(s_r),
    \]
    each of which form (trivial) broken cycles, yielding forbidden subrelations:
    \[
        W_1 = \big\{(s_b, s_r)\big\} \quad S_1 = \varnothing
    \]
    and
    \[
        W_2 = \big\{(s_r, s_b)\big\} \quad S_2 = \varnothing,
    \]
    and whose collapse is precisely $\langle \varnothing, \varnothing \rangle$. Thus the data fail to be strongly acyclic, and hence are inconsistent with any qualitative probability.\hfill $\blacksquare$
\end{example}

\section{Out-of-Sample Predictions}

In light of \autoref{generalchar}, the collapse operation provides a purely algorithmic means of evaluating whether or not an $\mathcal{M}$-invariant rationalizing preference exists for the data.  However, the collapse is defined over sets of restrictions, rather than the data itself. In particular, it does not speak to which comparisons \emph{every} $\mathcal{M}$-invariant, rationalizing preference must agree upon. When the data are rationalizable by at least one such preference, we term these comparisons the {\bf out-of-sample predictions} generated by the model and data.\medskip

When $\mathcal{M} = \{\textrm{id}\}$, every ($\mathcal{M}$-invariant) rationalizing preference $\succeq^*$ ranks $x \succeq^* y$ if and only if $x \succsim^R_\intercal y$.  However, as illustrated by \autoref{knockonexample}, when $\mathcal{M}$ is richer, so too are the set of counterfactual predictions generated by the class of $\mathcal{M}$-invariant preferences. Moreover, \autoref{knockonexample} shows that the set of such predictions is richer than either the transitive, or $\mathcal{M}$-invariant closure.\medskip

It turns out, however, that the set of out-of-sample predictions generated by the $\mathcal{M}$-invariant rationalizations of a strongly acyclic data set $\langle \succsim^R, \succ^R\rangle$ are straightforwardly described by collapses.

\begin{theorem}\label{invariantdm}
    Suppose $\langle \succsim^R ,\succ^R\rangle$ is strongly acyclic.  Then $x \succeq^* y$ for every $\mathcal{M}$-invariant rationalization $\succeq^*$ if and only if:
    \[
        \big\langle (y, x), (y,x)\big \rangle \in \mathcal{F}^*,
    \]
    and $x \succ^* y$ for every such rationalization if and only if:
    \[
        \big\langle (y,x), \varnothing \big\rangle \in \mathcal{F}^*.
    \]
\end{theorem}

\noindent If $\big\langle (y, x), (y,x)\big \rangle \in \mathcal{F}^*$, then there is a constraint on the extension problem requiring no rationalization to rank $y \succ^* x$. In this case, every rationalization must rank $x \succeq^* y$.\footnote{Analogously, $\big\langle (y,x), \varnothing \big\rangle \in \mathcal{F}^*$ encodes the restriction that no rationalization can rank $y \succeq^* x$ and hence every rationalization must rank $x \succ^* y$.}  As such, this condition is clearly necessary.  However, the primary content of \autoref{invariantdm} is that look only at such `singleton' restriction sets is also sufficient: the set of comparisons forced by restrictions of this form are, in fact, the \emph{only} comparisons agreed upon by every rationalization.  This provides a complete solution to (\hyperlink{q2}{Q.2}), for any choice of invariance axioms, data, or domain.

\section{Conclusion}

This paper studies the problem of characterizing the empirical content of structured families of preferences, satisfying axioms beyond rationality alone.  The basic observation underlying our results is that many of the most economically important and widely used decision-theoretic axioms share a common mathematical structure: they are what we have termed `invariance axioms.' Our main results provide characterizations of the empirical content and out-of-sample predictions generated by \emph{arbitrary} sets of such axioms.  The advantage of this abstraction is that it provides a unified theory and framework for studying a wide range of seemingly disparate economic models that had previously only been studied in isolation. By clarifying the common underlying structure at play, we hope that further work may build on the results here to develop further `universal' revealed preference characterizations.\medskip

\pagebreak 

\section*{Appendix}
\begin{appendix}
    \section{Proof of \autoref{commutativechar}}

Recall that a binary relation relation $\succeq \; \subseteq X \times X$ is $\mathcal{M}$-{\bf invariant} if:
\[
    x \succeq y \implies \omega(x) \succeq \omega(y)
\]
and
\[
    x \succ y \implies \omega(x) \succ \omega(y),
\]
for all $x,y \in X$ and $\omega \in \mathcal{M}$.  We say that $\succeq$ is {\bf strongly} $\mathcal{M}$-invariant if:
\[
    x \succeq y \iff \omega(x) \succeq \omega(y)
\]
and
\[
    x \succ y \iff \omega(x) \succ \omega(y).
\]
\begin{lemma}\label{lem:transitive}Suppose that an acyclic relation $\succsim^R$ is $\mathcal{M}$-invariant.  Then so is its transitive closure, $\succsim^R_\intercal$.\end{lemma}

\begin{proof}  First, let $x,y \in X$ such that $x\succsim^R_\intercal y$.  Then there exists $x_1, \ldots, x_K \in X$, $K \ge 2$, such that $x = x_1 \succsim^R \cdots \succsim^R x_K =y$.  By $\mathcal{M}$-invariance of $\succsim^R$, for every $\omega \in \mathcal{M}$, we also have that $\omega(x)=\omega(x_1) \succsim^R \cdots \succsim^R \omega(x_k) = \omega(y)$, hence $\omega(x) \succsim^R_\intercal \omega(y)$ as desired.  \medskip

Now, suppose that $x\succsim^R_\intercal y$ but it is not the case that $y\succsim^R_\intercal x$.  We want to show that it is not the case that $\omega(y) \succsim^R_\intercal \omega(x)$ for any $\omega \in \mathcal{M}$. As $x \succsim^R_\intercal y$, there exist $x_1,\ldots, x_K \in X$, $K \ge 2$ such that $x= x_1 \succsim^R \cdots \succsim^R x_K = y$; since additionally it is not the case that $y \succsim^R_\intercal x$,  for some $1 \le i \le K-1$, we have $x_i \succ x_{i+1}$.  Consequently, for any $\omega \in \mathcal{M}$, by 
 $\mathcal{M}$-invariance, we must also have $\omega(x) = \omega(x_1) \succsim^R \cdots \succsim^R \omega(x_K) = \omega(y)$, where $\omega(x_i) \succ \omega(x_{i+1})$.  By acyclicity of $\succsim^R$, it is then not the case that $\omega(y) \succsim^R_\intercal \omega(x)$, and hence $\omega(x) \succ^R_\intercal \omega(y)$ as desired. As $\omega\in \mathcal{M}$ was arbitrary, the result follows.
 \end{proof}


\begin{lemma}\label{lem:strongcoherent2}Suppose $\mathcal{M}$ is a commutative family.  Then every $\mathcal{M}$-invariant preorder has a strongly $\mathcal{M}$-invariant preorder extension.  
\end{lemma}

\begin{proof} 
Let $\succsim^R$ be a weakly $\mathcal{M}$-invariant preorder.  Define $\succeq$ via $x \succeq y$ if and only if there exists $\omega\in \mathcal{M}$ such that $\omega(x) \succsim^R \omega(y)$.\footnote{Recall that as $\mathcal{M}$ is closed under composition, this is equivalent to the existence of $\omega_1,\ldots, \omega_K \in \mathcal{M}$ such that $(\omega_1\circ\cdots \circ \omega_K)(x) \succsim^R (\omega_1\circ \cdots \circ \omega_K)(y)$.} \medskip

Since the identity function $\textrm{id} \in \mathcal{M}$, it follows immediately that $\succsim^R \; \subseteq\;  \succeq$. Suppose, now, that $x \succ^R y$ and, for purposes of contradiction that additionally $y \succeq x$. Then there exist $\omega \in \mathcal{M}$ for which $\omega(y) \succsim^R \omega(x)$. Since $\succsim^R$ is $\mathcal{M}$-invariant, this implies that both: $\omega(y) \succsim^R \omega(x)$ and $\omega(x) \succ^R \omega(y)$, a contradiction. Thus $\succ^R \;\subseteq \, \; \succ$ as well, and hence $\succeq$ defines an extension of $\succsim^R$.\medskip

We now claim that $\succeq$ is transitive.  Suppose that $x \succeq y \succeq z$.  As $x\succeq y$, there exist $\omega \in \mathcal{M}$ for which $\omega(x) \succsim^R \omega(y)$.  Similarly, as $y \succeq z$, there exist $\omega'\in \mathcal{M}$ for which $\omega'(y)\succeq \omega'(z)$.  By the $\mathcal{M}$-invariance of $\succsim^R$ and by commutativity of $\mathcal{M}$, we obtain $(\omega \circ \omega')(x)\succsim^R (\omega \circ \omega')(y)$ and $(\omega \circ \omega')(y)\succsim^R (\omega \circ \omega')(z)$, and hence $(\omega \circ \omega')(x)\succsim^R (\omega \circ \omega')(z)$ by transitivity of $\succsim^R$. But this means $x \succeq z$, thus we conclude $\succeq$ is transitive.\medskip

We now show that $\succeq$ is $\mathcal{M}$-invariant.  Suppose that $x \succeq y$ and let $\bar{\omega}\in \mathcal{M}$.  There exists $\omega\in M$ for which $\omega(x) \succsim^R \omega(y)$.  By commutativity of $\mathcal{M}$ and the $\mathcal{M}$-invariance of $\succsim^R$, we have $(\omega \circ \bar{\omega})(x) \succeq (\omega \circ \bar{\omega})(y)$, and hence $\omega(x) \succeq \omega(y)$.  Suppose now, additionally, that $x \succ y$ and, for sake of contradiction that for some $\omega'\in \mathcal{M}$, $\omega'(y) \succeq \omega'(x)$.  Then there exists $\omega''\in \mathcal{M}$ for which $(\omega'' \circ \omega')(y) \succsim^R (\omega'' \circ \omega')(x)$, which by definition implies that $y  \succeq x$, a contradiction. Hence $\succeq$ is $\mathcal{M}$-invariant.\medskip

Finally, we show that $\succeq$ is strongly $\mathcal{M}$-invariant. Suppose that $\omega(x) \succeq \omega(y)$. Then there exist $\omega' \in \mathcal{M}$ for which $(\omega' \circ \omega)(x) \succsim^R (\omega' \circ \omega)(y)$, which implies $x \succeq y$.  Suppose further $\omega(y) \succeq \omega(x)$ is false but, for sake of contradiction, that $y \succeq x$.  Then there exist $\omega'' \in \mathcal{M}$ such that $\omega''(y)\succeq \omega''(x)$.  By $\mathcal{M}$-invariance of $\succsim^R$, $(\omega \circ \omega'')(y) \succsim^R (\omega \circ \omega'')(x)$.  But then by the commutativity of $\mathcal{M}$, we conclude $\omega(y) \succeq \omega(x)$, a contradiction.  The result follows. \end{proof}

 \begin{lemma}\label{lem:acyclic2}Let $\mathcal{M}$ be a commutative family.  Let $\succeq$ be an $\mathcal{M}$-invariant preorder, and $w,z\in X$ be $\succeq$-unrelated (and hence distinct) elements of $X$. Then there is an acyclic $\mathcal{M}$-invariant extension $\succeq’$ of $\succeq$ that renders $w$ and $z$ comparable.\end{lemma}
 
 \begin{proof}
 For each $\omega \in \mathcal{M}$, let $e_\omega : \mathcal{M} \to \mathbb{Z}$ denote the function satisfying $\omega \mapsto 1$ and $\omega' \mapsto 0$ for all $\omega \neq \omega'$. By commutativity, any finite string of compositions of functions in $\mathcal{M}$ may be associated with a finitely-supported, non-negative valued, function $\mathcal{M} \to \mathbb{Z}$ via:
\[
    \omega_1^{n_1} \circ \omega_2^{n_2} \circ \cdots \circ \omega_K^{n_K} \quad \mapsto \quad n_1 e_{\omega_1} + \cdots + n_K e_{\omega_K},
\]
where $\omega^n$ denotes the $n$-fold composition of $\omega$ with itself. Let $\mathcal{M}^*$ denote the set of all such functions; conversely, every element of $\mathcal{M}^*$ clearly corresponds to some (composition of elements in $\mathcal{M}$ and hence) element of $\mathcal{M}$. Note that if $\mathbf{f}, \mathbf{g} \in \mathcal{M}^*$ represent finite strings of transformations in $\mathcal{M}$, then $\mathbf{f} + \mathbf{g}$ represents their composition.  For the remainder of this proof, we will freely associate elements of $\mathcal{M}$ with some fixed choice of representative in $\mathcal{M}^*$; the non-uniqueness of this selection will be irrelevant.\medskip 
 
 Suppose now, for sake of obtaining a contradiction, that no acyclic, $\mathcal{M}$-invariant extension of $\succeq$ exists that compares $w$ and $z$.  Thus every $\mathcal{M}$-invariant binary relation that extends $\succeq$ and renders $w$ and $z$ comparable, contains some cycle; in particular, the minimal such extensions obtained either (i) by adding $w \succ' z$ and $\mathbf{f}(w) \succ' \mathbf{f}(z)$ for all $\mathbf{f}$ associated with some finite composition of elements of $\mathcal{M}$, (ii) by adding $z \succ' w$ and all $\mathbf{f}(z) \succ' \mathbf{f}(w)$, or (iii) by adding $z \sim' w$ and all $\mathbf{f}(z) \sim' \mathbf{f}(w)$, must contain some cycle.  Consider first $\succeq' \, = \; \succeq  \cup  \succeq^*$, where $\succeq^*$ contains all relations of the form $w \succ^* z$ and $\mathbf{f}(w) \succ^* \mathbf{f}(z)$ for all finite compositions of elements of $\mathcal{M}$, $\mathbf{f}$.  Since $\succeq$ is a preorder, it follows there exists a cycle in $\succeq'$ composed of relations of two forms:
\begin{equation}\label{cycle1}
\begin{aligned}
\mathbf{a}^1(z) \succeq \mathbf{a}^2(w) & \quad & \mathbf{a}^2(w) \succ^*_\intercal \mathbf{a}^2(z)\\
 \vdots \quad \quad \quad \quad & \quad &  \vdots \quad \quad \quad \quad \\
 \mathbf{a}^{I-1}(z) \succeq \mathbf{a}^I(w) & \quad & \mathbf{a}^I(w) \succ^*_\intercal \mathbf{a}^I(z)\\
\mathbf{a}^I(z) \succeq \mathbf{a}^1(w)  & \quad & \mathbf{a}^1(w) \succ^*_\intercal \mathbf{a}^1(z) 
\end{aligned}
\end{equation}
for some $x \in X$, where the left column consists of relations in $\succeq$ and the right sequences solely of relations in $\succeq' \setminus \succeq$.  Note that $I \ge 2$, and without loss of generality, each $\mathbf{a}^i$ is distinct.\footnote{If $I=1$, then we have $\mathbf{a}^1(z) \succeq x$ and $x \succeq \mathbf{a}^1(w)$, hence $\mathbf{a}^1(z)\succeq \mathbf{a}^1(w)$. Since $\succeq$ is $\mathcal{M}$-invariant, this would imply $w$ and $v$ are $\succeq$-related, which is false.}\medskip

Analogously, if $\succeq' \; = \; \succeq \; \cup \; \succeq^*$, where $\succeq^*$ contains all relations of the form $z \succ^* w$ and $\mathbf{f}(z) \succ^* \mathbf{f}(w)$ for finite compositions $\mathbf{f}$, then there exists a cycle of the form:
\begin{equation}\label{cycle2}
\begin{aligned}
\mathbf{b}^1(w) \succeq \mathbf{b}^2(z) & \quad & \mathbf{b}^2(z) \succ^*_\intercal \mathbf{b}^2(w)\\
 \vdots \quad \quad \quad \quad & \quad &  \vdots \quad \quad \quad \quad \\
 \mathbf{b}^{J-1}(w) \succeq \mathbf{b}^J(z) & \quad & \mathbf{b}^J(z) \succ^*_\intercal \mathbf{b}^J(w)\\
\mathbf{b}^J(w) \succeq  \mathbf{b}^1(z) & \quad & \mathbf{b}^1(z) \succ^*_\intercal \mathbf{b}^1(w)
\end{aligned}
\end{equation}
for some $x' \in X$, where again the left column consists of relations in $\succeq$, the right solely of sequences of relations in $\succeq' \setminus \succeq$, $J \ge 2$, and each $\mathbf{b}^j$ unique.\medskip 

Finally, suppose $\succeq' \; = \; \succeq \; \cup \; \succeq^*$, where $\succeq^*$ contains all relations of the form $z \sim^* w$ and $\mathbf{f}(z) \sim^* \mathbf{f}(w)$ for finite compositions $\mathbf{f}$.  By hypothesis, there is a cycle of the form:
 \begin{equation}\label{cycle3}
\begin{aligned}
\mathbf{c}^1(y_1) \succeq \mathbf{c}^2(x_2) & \quad & \mathbf{c}^2(x_2) \sim^*_\intercal \mathbf{c}^2(y_2)\\
 \vdots \quad \quad \quad \quad & \quad &  \vdots \quad \quad \quad \quad \\
 \mathbf{c}^{K-1}(y_{K-1}) \succeq \mathbf{c}^K(x_K) & \quad & \mathbf{c}^K(a_K) \sim^*_\intercal \mathbf{c}^K(y_K)\\
\mathbf{c}^K(y_K) \succeq  \mathbf{c}^1(x_1) & \quad & \mathbf{c}^1(x_1) \sim^*_\intercal \mathbf{c}^1(y_1)
\end{aligned}
\end{equation}
where at least one relation in the left-hand column is strict, $K \ge 2$, each $\mathbf{c}^k$ is unique, and for all $k = 1,\ldots, K$, $\{x_k,y_k\} = \{w,z\}$.\medskip

Now, define:
\begin{equation*} 
\begin{aligned}
\mathbf{p}^i & = \mathbf{a}^{i+1} - \mathbf{a}^i\\
\mathbf{q}^j & = \mathbf{b}^{j+1} - \mathbf{b}^j\\
\mathbf{r}^k & = \mathbf{c}^{k+1} - \mathbf{c}^k,\\
\end{aligned}
\end{equation*}
where we interpret indices $I + 1, J+1, K+1 \equiv 1$.  Note that each $\mathbf{p}^i, \mathbf{q}^j,$ and  $\mathbf{r}^k$ is not equal to the zero function $\mathbf{0}$ and, by construction:
\begin{equation*}
    \sum_{i =1}^I \mathbf{p}^i  = \sum_{j =1}^J \mathbf{q}^j = \sum_{k =1}^K \mathbf{r}^k  =  \mathbf{0}.
\end{equation*}
Consider the sets:
\begin{equation*}
\begin{aligned}
    \tilde{A}_{wz} & = \big\{\mathbf{r}^k \; \vert \; y_k = w, \; x_{k+1} = z\big\} \\
    \tilde{A}_{zw} & = \big\{\mathbf{r}^k \; \vert \; y_k = z, \; x_{k+1} = w\big\} \\
    \tilde{A}_{ww} & = \big\{\mathbf{r}^k \; \vert \; y_k = w, \; x_{k+1} = w\big\} \\
    \tilde{A}_{zz} & = \big\{\mathbf{r}^k \; \vert \; y_k = z, \; x_{k+1} = z\big\}.
\end{aligned}
\end{equation*}
Clearly these sets cover $\{\mathbf{r}^1, \ldots, \mathbf{r}^K\}$. However, they may not define a partition. Thus let:
\begin{equation*}
\begin{aligned}
    A_{wz} & = \tilde{A}_{wz}\\
    A_{zw} & = \tilde{A}_{zw} \setminus \tilde{A}_{wz}\\
    A_{ww} & = \tilde{A}_{ww} \setminus \tilde{A}_{zw} \setminus \tilde{A}_{wz}\\
    A_{zz} &= \tilde{A}_{zz} \setminus \tilde{A}_{ww} \setminus \tilde{A}_{zw} \setminus \tilde{A}_{wz},
\end{aligned}
\end{equation*}
if these sets are non-empty, and otherwise define them as $\{\mathbf{0}\}$.  By hypothesis, at least one of the $A$ sets must contain non-zero elements. Note that each element of  $\{\mathbf{r}^1, \ldots, \mathbf{r}^K\}$ is contained in exactly one set in the collection $\{A_{wz}, A_{zw}, A_{ww}, A_{zz}\}$.  Let $\{\mathbf{s}^m_{wz}\}_{m=1}^{\vert A_{wz}\vert}$ (resp. $\{\mathbf{s}^m_{zw}\}_{m=1}^{\vert A_{zw}\vert}$, $\{\mathbf{s}^m_{ww}\}_{m=1}^{\vert A_{ww}\vert}$, and $\{\mathbf{s}^m_{zz}\}_{m=1}^{\vert A_{zz}\vert}$) denote enumerations of $A_{wz}$ (resp. $A_{zw}, A_{ww},$ and $A_{zz}$).\medskip

We now establish a contradiction, by showing that $\succeq$ contains a cycle, contrary to our hypothesis that it is a preorder.  Let $\bar{\mathbf{h}}$ denote a sufficiently large vector in $\mathcal{M}^*$.\footnote{Sufficiently in the sense only that each vector in the following sequence remain non-negative valued.}  We will consider two cases in turn.\medskip

\noindent \textbf{Case 1:  $|A_{wz}|+|A_{zw}|>0$.}\medskip

To build our cycle, we first define two chains in $\succeq$ which will prove important in our construction.\footnote{The first chain indexes by $|A_{wz}|$ and the second indexes by $|A_{zw}|$; if either of these are zero, these chains are trivial.} By the top-left relation in \eqref{cycle1}, we have:
\[
    \mathbf{a^1}(z) \succeq \mathbf{a^2}(w).
\]
By $\mathcal{M}$-invariance, this implies:
\[
    (\bar{\mathbf{h}} + \mathbf{a^1})(z) \succeq (\bar{\mathbf{h}} + \mathbf{a^2})(w),
\]
and hence, so long as $\bar{\mathbf{h}}$ is large enough, i.e. $\bar{\mathbf{h}} + \mathbf{p}^1 \ge \mathbf{0}$, we have:
\[
    (\bar{\mathbf{h}})(z) \succeq (\bar{\mathbf{h}} + \mathbf{p}^1)(w),
\]
by (full) $\mathcal{M}$-invariance. By repeating this logic, and also applying it to relations from \eqref{cycle2} and \eqref{cycle3}, we can obtain lengthy chains of $\succeq$-relations. Let us refer to chain one as the sequence:

\begin{equation*}
    \begin{aligned}
       \bar{\mathbf{h}}(z) & \succeq (\bar{\mathbf{h}} + \mathbf{p}^1)(w) \\
       & \succeq (\bar{\mathbf{h}} + \mathbf{p}^1 + \mathbf{s}^1_{wz})(z) \\
       & \quad \quad \quad \quad \quad \vdots \\
       & \succeq \bigg(\bar{\mathbf{h}} + \vert A_{wz} \vert \sum_{i=1}^I \mathbf{p}^i + I \sum_{m=1}^{\vert A_{wz}\vert} \mathbf{s}^m_{wz}\bigg)(z) \\
       & \quad \quad \quad \quad \quad \vdots \\
       & \succeq \bigg(\bar{\mathbf{h}} + J \, \vert A_{wz} \vert \sum_{i=1}^I \mathbf{p}^i + IJ \sum_{m=1}^{\vert A_{wz}\vert} \mathbf{s}^m_{wz}\bigg)(z).
    \end{aligned}
\end{equation*}
which follows simply by repeating application of the above observation.\footnote{The first part of this chain, up to: \[
\bigg(\bar{\mathbf{h}} + \vert A_{wz} \vert \sum_{i=1}^I \mathbf{p}^i + I \sum_{m=1}^{\vert A_{wz}\vert} \mathbf{s}^m_{wz}\bigg)(z)
\]
is constructed as follows: for every $l=1,\ldots,I|A_{wz}|$, every term of the form $(\bar{\mathbf{h}}+\ldots +\mathbf{p}^l)(w)$ is followed by a term of the form $(\bar{\mathbf{h}}+\ldots+\mathbf{p}^l +\mathbf{s}_{wz}^l)(z)$, and for every $l=0,\ldots,I|A_{wz}|-1$, every term of the form $(\bar{\mathbf{h}}+\ldots +\mathbf{s}_{wz}^l)(z)$ is followed by a term of the form $(\bar{\mathbf{h}}+\ldots+\mathbf{p}^l +\mathbf{s}_{wz}^{l+1})(w)$, where indices on $\mathbf{p}$ are to be understood modulo $I$ and on $s_{wz}$ modulo $|A_{wz}|$ as above. The second part of this chain, up through: 
\[
\bigg(\bar{\mathbf{h}} + J \, \vert A_{wz} \vert \sum_{i=1}^I \mathbf{p}^i + IJ \sum_{m=1}^{\vert A_{wz}\vert} \mathbf{s}^m_{wz}\bigg)(z),
\]
follows by iterating the first $I|A_{wz}|$ steps of this construction an additional $|J|-1$ times.} Similarly, we refer to chain two as the sequence of relations:
\begin{equation*}
    \begin{aligned}
       \bar{\mathbf{h}}(z) & \succeq (\bar{\mathbf{h}} + \mathbf{s}^1_{zw})(w)\\
       & \succeq (\bar{\mathbf{h}} + \mathbf{s}^1_{zw}+ \mathbf{q}^1)(z) \\
       & \quad \quad \quad \quad \quad \vdots \\
       & \succeq \bigg(\bar{\mathbf{h}} + J \sum_{m=1}^{\vert A_{zw}\vert} \mathbf{s}^m_{wz} + \vert A_{zw} \vert \sum_{j=1}^J \mathbf{q}^j \bigg)(z) \\
       & \quad \quad \quad \quad \quad \vdots \\
       & \succeq \bigg(\bar{\mathbf{h}}  + I J \sum_{m=1}^{\vert A_{zw}\vert} \mathbf{s}^m_{wz} + I \, \vert A_{zw}  \vert \sum_{j=1}^J \mathbf{q}^j\bigg)(z).
    \end{aligned}
\end{equation*}

\noindent Appending these chains together then yields a chain:
\begin{equation*}
\bar{\mathbf{h}}(z) \succeq \cdots \succeq \bigg(\bar{\mathbf{h}} + I \, \vert A_{zw}  \vert \sum_{j=1}^J \mathbf{q}^j  + J \, \vert A_{wz} \vert \sum_{i=1}^I \mathbf{p}^i + IJ \sum_{m=1}^{\vert A_{wz}\vert} \mathbf{s}^m_{wz} + I J \sum_{m=1}^{\vert A_{zw}\vert} \mathbf{s}^m_{wz}\bigg)(z).
\end{equation*}

Consider now the following modification to this chain: immediately after the first instance of an $\mathbf{f}(z) \succeq \mathbf{g}(w)$ relation, apply $IJ$ applications of each transformation in $A_{ww}$.  Similarly, after the first $\mathbf{f}(w) \succeq \mathbf{g}(z)$ relation, insert $IJ$ repetitions of each transformation in $A_{zz}$.  The result is a chain:

\begin{equation*}
\bar{\mathbf{h}}(z) \succeq \cdots \succeq \bigg(\bar{\mathbf{h}} + I \, \vert A_{zw}  \vert \sum_{j=1}^J \mathbf{q}^j  + J \, \vert A_{wz} \vert \sum_{i=1}^I \mathbf{p}^i + IJ \sum_{k=1}^K \mathbf{r}^k\bigg)(z).
\end{equation*}

However, since $\sum_i \mathbf{p}^i= \sum_j \mathbf{q}^j = \sum_k \mathbf{r}^k = \mathbf{0}$, the first and last terms in this chain coincide.  Moreover, since every relation in the left-hand column of \eqref{cycle3} appears in this cycle, the sequence contains at least one strict relation, contradicting the hypothesis that $\succ$ is a preorder.\medskip

\noindent \textbf{Case 2:  $|A_{wz}|+|A_{zw}|=0$.}\medskip

Here, we follow a similar construction to the preceding case, except here we first consider a single chain of the form:

\begin{equation*}
    \begin{aligned}
       \bar{\mathbf{h}}(z) & \succeq (\bar{\mathbf{h}} + \mathbf{p}^1)(w) \\
       & \succeq (\bar{\mathbf{h}} + \mathbf{p}^1 + \mathbf{q}^1)(z) \\
       & \quad \quad \quad \quad \quad \vdots \\
       & \succeq \bigg(\bar{\mathbf{h}} + J  \sum_{i=1}^I \mathbf{p}^i + I \sum_{j=1}^{J} \mathbf{q}^j\bigg)(z). 
    \end{aligned}
\end{equation*}

Consider now the following modification to this chain: immediately after the first instance of an $\mathbf{f}(z) \succeq \mathbf{g}(w)$ relation, insert one application of each transformation in $A_{ww}$.  Similarly, after the first $\mathbf{f}(w) \succeq \mathbf{g}(z)$ relation, insert an application of each transformation in $A_{zz}$.  The result is a chain:

\begin{equation*}
\bar{\mathbf{h}}(z) \succeq \cdots \succeq \bigg(\bar{\mathbf{h}} + I \sum_{j=1}^J \mathbf{q}^j  + J  \sum_{i=1}^I \mathbf{p}^i + \sum_{k=1}^K \mathbf{r}^k\bigg)(z).
\end{equation*}

But by analogous logic to the former case, this also defines a cycle, contradicting the assumption that $\succeq$ is a preorder. Since these cases are exhaustive, we conclude such an extension must exist, which completes the proof.
\end{proof}

\noindent We now are in a position to prove \autoref{commutativechar}.
\begin{proof}

Suppose $\succsim^R_\mathcal{M}$ is acyclic.  By \autoref{lem:transitive}, the transitive closure of $\succsim^R_\mathcal{M}$ is an $\mathcal{M}$-invariant pre-order, and hence by \autoref{lem:strongcoherent2} admits a strongly $\mathcal{M}$-invariant preorder extension.\medskip

The remainder of the proof follows from a standard transfinite induction argument.  Let $\mathcal{P}_\mathcal{M}$ denote the set of strongly $\mathcal{M}$-invariant preorders on $X$, partially ordered by extension. Given $\succeq_1, \succeq_2\, \in \mathcal{P}_\mathcal{M}$, we write $\succeq_1\rhd \succeq_2$ whenever $\succeq_1$ extends $\succeq_2$. Let $\{\succeq_\lambda\}_{\lambda \in \Lambda}$ be an arbitrary $\rhd$-chain of $\mathcal{M}$-invariant preorders. It follows from standard arguments (see, e.g., \citealt{richter1966revealed, chambers2016revealed}) that:
\[
    \bar{\succeq} = \bigcup_{\lambda \in \Lambda} \succeq_\lambda
\]
is a preorder extension of every $\succeq_\lambda$. Similarly, it follows  that $\bar{\succeq}$ is strongly $\mathcal{M}$-invariant: if $(x,y) \in \bar{\succeq}$, then there exists some $\lambda \in \Lambda$ such that $x \succeq_\lambda y$, and since $\succeq_\lambda$ is strongly $\mathcal{M}$-invariant, so must be $\bar{\succeq}$ since it extends $\succeq_\lambda$. Hence $\bar{\succeq}$ belongs to $\mathcal{P}_\mathcal{M}$, and by Zorn's Lemma, there exists a maximal strongly $\mathcal{M}$-invariant preorder $\succeq^*$ which extends $\succsim^R_\mathcal{M}$. Suppose, for purposes of obtaining a contradiction, that $\succeq^*$ is not complete. Then there exist $w,z \in X$ that are $\succeq^*$-unrelated. By \autoref{lem:acyclic2} there exists a strongly $\mathcal{M}$-invariant preorder extension of $\succeq^*$ that renders $w$ and $z$ comparable, however, this contradicts the $\rhd$-maximality of $\succeq^*$.  Thus $\succeq^*$ is complete and hence is an $\mathcal{M}$-invariant rationalizing preference for $\succsim^R_\mathcal{M}$, and hence $\succsim^R$.
\end{proof}

\section{Relating \autoref{commutativechar} and GARP Variations: The Case of Price-Consumption Data}\label{garpapp}

In this section, we consider the special case in which our relations $\langle \succsim^R, \succ^R\rangle $ are generated by some price-consumption data set $\{(p_1,x_1),\ldots, (p_K, x_K)\}$.  Here, we assume $\langle \succsim^R, \succ^R\rangle $ are the revealed preference relations associated with this data set, via:
\[
    x \succsim^R y \quad \iff \quad x \, = \, x^k \textrm{ for some $k$, and } p_k \cdot x \ge p_k \cdot y
\]
(respectively $\succ^R$ and $>$). We show that for various common choices of $\mathcal{M}$, the acyclicity of $\langle \succsim^R_\mathcal{M}, \succ^R_\mathcal{M}\rangle$ straightforwardly reduces to the standard, model-specific revealed preference axioms.\medskip

\subsubsection{Quasilinearity}
    Suppose that $X = Y \times \mathbb{R}_+$, and $\mathcal{M}$ consists of all transformations of the form $(y, t) \mapsto (y, t+ \alpha)$ for $\alpha \ge 0$.  Let $\langle \succsim^R, \succ^R \rangle$ be an arbitrary data set. Then the $\mathcal{M}$-closure $ \langle \succsim^R_\mathcal{M}, \succ^R_\mathcal{M}\rangle$ is defined by:
    \[
        (y,t+\alpha) \succsim^R_\mathcal{M} (y',t'+\alpha) \quad \iff \quad (y, t) \succsim^R (y' , t')
    \]
    for some $\alpha \ge 0$, and analogously for $\succ^R_{\mathcal{M}}$.\medskip
    
    Suppose now that $Y = \mathbb{R}^{L-1}_+$, and $\langle \succsim^R , \succ^R \rangle$ is the revealed preference relation arising from some price-consumption data set; without loss of generality, we normalize each $p_k = (\tilde{p}_k, 1)$. Then a $\langle\succsim^R_\mathcal{M}, \succ^R_\mathcal{M}\rangle$ cycle is equivalent to the existence of $(y_{k_0}, t_0) , \ldots, (y_{k_N}, t_N) \in X$, and $\alpha_0, \ldots, \alpha_N \ge 0$ such that:
    \begin{equation}\label{quasilinearcycle}
        \begin{aligned}
            p_{k_0} \cdot (y_{k_0}, t_{k_0}) & \ge p_{k_0} \cdot (y_{k_1},t_1 + \alpha_0) = p_{k_0} \cdot (y_{k_1},t_{k_1} + \alpha_0 - \alpha_1)\\
            p_{k_1} \cdot (y_{k_1}, t_{k_1}) & \ge p_{k_1}\cdot (y_{k_2},t_2+ \alpha_1) = p_{k_1} \cdot (y_{k_2},t_{k_2} + \alpha_1 - \alpha_2)\\
            & \; \; \vdots \\
            p_{k_N} \cdot (y_{k_N}, t_{k_N}) & > p_{k_N} \cdot (y_{k_0},t_0 + \alpha_N) = p_{k_N} \cdot (y_{k_0},t_{k_0} + \alpha_N - \alpha_0)\\
        \end{aligned}
    \end{equation}  
    where $t_{k_i} = t_i + \alpha_i$ for all $i = 1,\ldots N$.\footnote{In other words, $x\succsim^R_\mathcal{M} y$ if and only if there is some fixed translation along the numeraire axis that brings $x$ equal to some chosen $x_k$, and which leaves $y$ within the budget defined by $p_k$ and $x_k$.} Summing over (\ref{quasilinearcycle}):
    \[
        \begin{aligned}
            \sum_{i=0}^N \tilde{p}_{k_i} \cdot (y_{k_{i+1}} - y_{k_{i}}) < 0,
        \end{aligned}
    \]
    which precisely corresponds precisely to a negative cycle \`{a} la \cite{brown2007nonparametric}.\footnote{Here, the $i$ indices are understood to satisfy $N+1 \equiv 0$.} 

\subsubsection{Homotheticity}
    Let $X$ be a cone in a real vector space, and let $\mathcal{M}$ consist of all transformations of the form $x \mapsto \alpha x$, for $\alpha > 0$. The particular case of $X=\mathbb{R}^n_+$ is treated in \citet{chambers2016revealed}, Theorem 4.2, but we reproduce the ideas here.  
    
    Here, the $\mathcal{M}$-closure of the data set $\langle \succsim^R, \succ^R\rangle$ is given by:
    \[
        x \succsim^R_\mathcal{M} y \quad \iff \quad  \alpha x \succsim^R \alpha y
    \]
    for some $\alpha > 0$, with a similar definition for $\succ^R_\mathcal{M}$.  In \citet{chambers2016revealed}, $\langle \succsim^R_\mathcal{M},\succ^R_\mathcal{M}\rangle$ is referred to as $\langle \succeq^H, \succ^H \rangle$.  The $\mathcal{M}$-closure is acyclic if and only if there do not exist $x_0,\ldots, x_N \in X$ and $\alpha_0,\ldots, \alpha_N > 0$ such that:
    \[
        \begin{aligned}
            \alpha_0 x_0 & \succsim^R \alpha_0 x_1\\
            \alpha_1 x_1 & \succsim^R \alpha_1 x_2\\
            & \; \; \vdots\\
            \alpha_N x_N & \succ^R \alpha_N x_0.
        \end{aligned}
    \]
    Suppose again that $\langle \succsim^R, \succ^R\rangle$ is the revealed preference relation arising from some set of price-consumption observations; without loss of generality, we normalize each price so $p_k \cdot x_k = 1$. Then (\ref{homotheticcycle}) is equivalent to the existence of $x_{k_0},\ldots, x_{K} \in X$ and $\alpha_0,\ldots, \alpha_K > 0$ such that:
    \begin{equation}\label{homotheticcycle}
        \begin{aligned}
            p_{k_0} \cdot x_{k_0} & \ge p_{k_0} \cdot (\alpha_0 x_1) = p_{k_0} \cdot \bigg(\frac{\alpha_0 x_{k_1}}{\alpha_1}\bigg)\\
            p_{k_1} \cdot x_{k_1} & \ge p_{k_0} \cdot (\alpha_1 x_2) = p_{k_0} \cdot \bigg(\frac{\alpha_1 x_{k_2}}{\alpha_2}\bigg)\\
            & \; \; \vdots \\
            p_{k_N} \cdot x_{k_N} & > p_{k_N} \cdot (\alpha_N x_0) = p_{k_N} \cdot \bigg(\frac{\alpha_N x_{k_0}}{\alpha_0}\bigg)
        \end{aligned}
    \end{equation}
    where $\alpha_i x_i = x_{k_i}$ for all $i = 1,\ldots, N$. Taking products of (\ref{homotheticcycle}) leads to the cancellations of all $\alpha_i/\alpha_{i+1}$ terms, resulting in:
    \[
        \prod_{i=0}^N p_{k_i} x_{k_{i+1}}< 1,
    \]
    which is precisely a violation of the homothetic axiom of revealed preference of \cite{varian1983non}. As mentioned previously, in the case of general $\langle \succsim^R , \succ^R \rangle$, not necessarily arising from price-consumption observations, \cite{demuynck2009} obtains a similar characterization, in the special case of monotone and homothetic preferences, via a different approach.

\subsubsection{Translation-Invariance}
    Let $S$ be some finite set of states of the world, and let $X = \mathbb{R}^S$ denote the space of portfolios of Arrow securities. Let $\mathcal{M}$ denote the collection of transformations of the form $x \mapsto x + \vec{\alpha}$, where $\vec{\alpha} := (\alpha,\ldots, \alpha)$, for each $\alpha \in \mathbb{R}$. We refer to an $\mathcal{M}$-invariant preference as \emph{translation invariant}.  By \autoref{commutativechar}, the data $\langle \succsim^R ,\succ^R\rangle$ are rationalizable by a translation-invariant preference if and only if there does not exist $x_0,\ldots, x_N \in X$ and $\alpha_0,\ldots, \alpha_N \in \mathbb{R}$ such that:
    \begin{equation}\label{transinvcycle}
        \begin{aligned}
            p_{k_0} \cdot x_{k_0} \ge p_{k_0} \cdot \big(x_1 + \vec{\alpha}_0\big) & = p_{k_0} \cdot \big(x_{k_1} + \vec{\alpha}_0 - \vec{\alpha}_1\big)\\
            p_{k_1} \cdot x_{k_1}  \ge p_{k_1} \cdot \big(x_2 + \vec{\alpha}_1\big) & = p_{k_1} \cdot \big(x_{k_2} + \vec{\alpha}_1 - \vec{\alpha}_2\big)\\
            & \; \;  \vdots \\
           p_{k_N} \cdot x_{k_N}  \ge p_{k_N} \cdot \big(x_0 + \vec{\alpha}_N\big) & = p_{k_N} \cdot \big(x_{k_0} + \vec{\alpha}_N - \vec{\alpha}_0\big).
        \end{aligned}
    \end{equation}
    Summing over (\ref{transinvcycle}) we obtain:
    \[
        \sum_{i=0}^N p_{k_i} \cdot (x_{k_{i+1}} - x_{k_i}) - \sum_{i=0}^N (\alpha_i - \alpha_{i+1}) \|p_{k_i}\|_1 < 0,
    \]
    or, normalizing each $p_{k_i}$ by $\|p_{k_i}\|_1$ without loss of generality:
    \[
        \sum_{i=0}^N \frac{p_{k_i}}{\|p_{k_i}\|_1} \cdot (x_{k_{i+1}} - x_{k_i}) < 0,
    \]
    precisely the same condition obtained in \cite{chambers2016testing}.\footnote{Economically, this normalization may be regarded as treating \emph{bonds} as a numeraire commodity.}

\section{Proof of \autoref{generalchar}}

\subsection{Preliminaries from Propositional Logic}

For all $(x,y) \in X \times X$, define two boolean variables:
\[
[\texttt{x} \succeq \texttt{y}] \quad \textrm{ and } \quad [\texttt{x} \succ \texttt{y}].
\]
Let $\mathcal{V}$ denote the set of all such variables.  A {\bf model} is a mapping $\mu: \mathcal{V} \to \{\top ,\bot\}$ assigning a truth value to every variable in $\mathcal{V}$.\footnote{In this appendix, we will exclusively use the word `model' in its logical interpretation, rather than its economic meaning in the main text.} We may extend any model from boolean variables to well-formed logical formulae in the obvious manner.  For a proof of this fact, and an introduction to propositional logic, the interested reader is referred to \cite{schoning2008logic}.\medskip

Every formula in propositional logic is equivalent to one in conjunctive normal form (CNF).\footnote{See \citealt{schoning2008logic}.} A {\bf literal} is an atomic formula, of the form $A$ or $\neg A$, for some $A \in \mathcal{V}$. A finite formula $F$ in conjunctive normal form can be written as:
\[
    F = (A_{1,1} \vee \cdots \vee A_{1,n_1}) \wedge \cdots \wedge (A_{K,1} \vee \cdots \vee A_{K, n_K}),
\]
where each $A_{i,j}$ is a literal. We view the formula $F$ as being formed by the individual {\bf clauses}:
\[
    C_i = A_{i,1} \vee \cdots \vee A_{i,n_i}.
\]
A formula such as $F$ can be compactly expressed in set notation:
\[
    \big\{\underbrace{\{A_{1,1}, \ldots, A_{1,n_1}\}}_{C_1} , \ldots, \underbrace{\{A_{K,1}, \ldots, A_{K,n_K}\}}_{C_K} \big\},
\]
where each $C_i = \{A_{i,1}, \ldots, A_{i,n_i}\}$ is a clause. In other words, within a clause, a comma denotes an OR operation (i.e.\ $\vee$), and a comma between clauses denotes an AND (i.e.\ $\wedge$). The formula consisting only of the empty clause $\{\varnothing\}$ is a valid formula; by definition it is unsatisfiable. \medskip

Let $C_1$, $C_2$, and $R$ be clauses. We say that $R$ is a {\bf resolvent} of $C_1$ and $C_2$ if there exists some literal $L$ such that $L \in C_1$ and $\neg{L} \in C_2$, and
\[
    R = \big(C_1 \setminus \{L\}) \; \cup \; (C_2 \setminus \{\neg L\}).
\]
\noindent The following property resolution is standard (see, e.g., \citealt{schoning2008logic} p.32).

\begin{lemma}\label{resolutionlemma}
    Let $\Theta$ be a set of clauses, and let $R$ be the resolvent of two clauses $C_1$ and $C_2$ in $\Theta$.  Then $\Theta$ and $\Theta \cup \{R\}$ are logically equivalent.
\end{lemma}

When speaking of resolvents, we explicitly allow for $R$ to be the empty set. Suppose $\Theta$ is a set of clauses.  A {\bf derivation} of $\varnothing$ via resolution is a finite sequence of clauses $\{C_1,\ldots, C_N\}$ such that:
\begin{itemize}
    \item[(i)] $C_N = \varnothing$; and
    \item[(ii)] For all $i = 1,\ldots, N$, $C_i$ is either a clause in $\Theta$, or a resolvent some $C_j$ and $C_k$ (the {\bf parents} of $C_i$), where $j,k < i$.
\end{itemize}
\noindent More generally, if we remove condition (i) we speak of a {\bf partial derivation} (of $C_N$). A set of clauses $\Theta$ is said to be {\bf unsatisfiable} if and only if there is no model which evaluates every formula in $\Theta$ to $\top$.  Remarkably, by forming a finite number of resolvents, one is always capable of detecting whether any finite set of formulas is unsatisfiable.

\begin{theorem}[\citealt{robinson1965machine}]\label{resolutionthm}
    Let $\Theta$ be a finite set of clauses. Then $\Theta$ is unsatisfiable if and only if there exists a derivation of $\varnothing$ via resolution.
\end{theorem}

\noindent The \cite{robinson1965machine} paper actually proves stronger analogous result, in the more general setting of first-order logic. For a proof of the above result in propositional logic, the interested reader is referred to \cite{schoning2008logic}, Chapter 1, Section 5.  Many refinements of \autoref{resolutionthm} exist, intended to further reduce the search space for proofs in the context of machine learning.  We will have use of the following modification: say a derivation $\{C_1,\ldots, C_N\}$ of $\varnothing$ is via {\bf negative resolution} if:
\begin{itemize}
    \item[(i)] $C_N = \varnothing$; and
    \item[(ii')] For all $i = 1,\ldots, N$, $C_i$ is either a clause in $\Theta$, or a resolvent of some $C_j$ and $C_k$, where $j,k < i$ and either $C_j$ or $C_k$ contains no positive literals.
\end{itemize}

\noindent The following theorem is proven on p.102 in \cite{schoning2008logic}.

\begin{theorem}\label{completenessnegres}
    Let $\Theta$ be a finite set of formulas. Then $\Theta$ is unsatisfiable if and only if there exists a derivation of $\varnothing$ via negative resolution.
\end{theorem}

\noindent \autoref{completenessnegres} provides a `representation theorem' for proofs of inconsistency: while there may be (many) proofs that a given set of clauses is unsatisfiable, \autoref{completenessnegres} guarantees that at least one can be carried out wholly via resolution where one parent at every step contains no positive literals. Crucially, every order pair in $\langle W, S\rangle \in \mathcal{F}^*$ uniquely defines a clause containing no positive literals:
\begin{equation}\label{clausalrep}
   \langle W,S \rangle \mapsto \bigg[\bigvee_{(x,y) \in W \setminus S} \neg [\texttt{x} \succeq \texttt{y}] \bigg] \vee \bigg[ \bigvee_{(x,y) \in S} \neg [\texttt{x} \succ \texttt{y}]\bigg].
\end{equation}
We term this the {\bf clausal representation} of the order pair $\langle W, S \rangle$. In particular, the clausal representation of $\langle \varnothing, \varnothing\rangle$ is the empty clause.

\subsection{$\mathcal{M}$-invariant Rationalization}

Let $\Phi$ denote the collection of all logical formulas of the following form:
\begin{itemize}
    \item[(T.1)] {\bf Completeness}: For all $x,y \in X$:
    \[
        [\texttt{x} \succeq \texttt{y}] \vee [\texttt{y} \succeq \texttt{x}].
    \]
    This is in conjunctive normal form (CNF).
    \item[(T.2)] {\bf Coherency}: For all $x,y \in X$:
    \[
        [\texttt{x} \succeq \texttt{y}] \iff \neg [\texttt{y} \succ \texttt{x}].
    \]
    In CNF, this may be regarded as two separate clauses,
    \[\tag{T.2.a}
        \neg [\texttt{x} \succeq \texttt{y}] \vee \neg [\texttt{y} \succ \texttt{x}]
    \]
    and
    \[\tag{T.2.b}
        [\texttt{x} \succeq \texttt{y}] \vee [\texttt{y} \succ \texttt{x}].
    \]
    \item[(T.3)] {\bf Transitivity}: For all $x,y,z \in X$:
    \[
        [\texttt{x} \succeq \texttt{y}] \wedge [\texttt{y} \succeq \texttt{z}] \implies [\texttt{x} \succeq \texttt{z}],
    \]
    or, in CNF:
    \[
        \neg [\texttt{x} \succeq \texttt{y}] \vee \neg [\texttt{y} \succeq \texttt{z}] \vee [\texttt{x} \succeq \texttt{z}].
    \]
    \item[(T.4)] {\bf Extension}: For all $(x,y) \in \; \succsim^R$,
    \[
        [\texttt{x} \succeq \texttt{y}].
    \]
    Moreover, if $(x,y) \in\; \succ^R$ then:
    \[
        [\texttt{x} \succ \texttt{y}].
    \] 
    \item[(T.5)] {\bf Invariance}: For all $x,y \in X$ and $\omega \in \mathcal{M}$ such that $x,y$ belong to the domain of $\omega$:
    \[
        [\texttt{x} \succeq \texttt{y}] \iff [\omega(\texttt{x}) \succeq \omega(\texttt{y})],
    \]
    or
    \[\tag{T.5.a}
        \neg [\texttt{x} \succeq \texttt{y}] \vee [\omega(\texttt{x}) \succeq \omega(\texttt{y})]
    \]
    and
    \[\tag{T.5.b}
        [\texttt{x} \succeq \texttt{y}] \vee \neg [\omega(\texttt{x}) \succeq \omega(\texttt{y})].
    \]
\end{itemize}
By construction, the set of models which evaluate to $\top$ for every formula in $\Phi$ are in 1-1 correspondence with the $\mathcal{M}$-invariant weak order extensions of $\langle \succsim^R , \succ^R\rangle$.\medskip

\subsection{Proofs}

We proceed in the proof of \autoref{generalchar} via several lemmas.

\begin{lemma}\label{resolutionsound}
    Suppose $\langle \varnothing, \varnothing\rangle \in \mathcal{F}^*$. Then there does not exist any $\mathcal{M}$-invariant preference relation extending $\langle \succsim^R, \succ^R\rangle$.
\end{lemma}
\begin{proof}
    By minor abuse of notation, we identify every order pair in $\mathcal{F}^*$ with its clausal representation under \eqref{clausalrep}. Let $\Theta$ denote the collection of all clauses of the form (T.1) - (T.5), as well as all clauses in $\mathcal{F}^0$; recall, $\mathcal{F}^0$ consists of the (clausal representations of) forbidden subrelations generated by broken cycles in the data.\medskip
    
    Suppose $\succeq$ is an $\mathcal{M}$-invariant weak order extension of $\langle \succsim^R, \succ^R \rangle$. By construction, the rule of assignment (i) $[\texttt{x} \succeq \texttt{y} ] = \top$ if and only if $x \succeq y$ and (ii) $[\texttt{x} \succ \texttt{y}]$ if and only if $x \succ y$, defines a valid model for $\Theta$, i.e. under these assignments, every clause in $\Theta$ evaluates to $\top$.\footnote{The clauses of the form (T.1)-(T.5) are clearly necessary as they define the basic properties of an invariant weak order extension. Clauses in $\mathcal{F}^0$ must also hold lest $\succeq$ contain a cycle. Every clause in $\mathcal{F}^0$ can be logically deduced from (T.1) - (T.5), however we do not need this fact in light of  standard order-theoretic arguments.} Thus $\Phi$ is satisfiable if and only if $\Theta$ is.\medskip
    
    Let $C,C' \in \mathcal{F}^0$ denote clausal representations of two forbidden subrelations, derived from broken cycles (either strict or weak) in the data, and suppose $D \in \mathcal{F}^1$ is the (clausal representation of the) collapse of $C$ and $C'$. Regarding these as sets of negative literals, there exists (negative) literals $L\in C$ and $L' \in C'$ such that:
    \[
        D = \big(C \setminus \{L\}\big) \; \cup \; \big(C' \setminus \{L'\}\big)
    \]
    and either:
\[
    L = \neg [\omega(\texttt{y}) \succeq \omega(\texttt{x})] \quad \textrm{ and } \quad L' = \neg [\omega'(\texttt{x}) \succeq \omega'(\texttt{y})]
\]
or
\[
    L = \neg [\omega(\texttt{y}) \succeq \omega(\texttt{x})] \quad \textrm{ and } \quad L' = \neg [\omega'(\texttt{x}) \succ \omega'(\texttt{y})]
\]
for some $x,y \in X$, $\omega,\omega' \in \mathcal{M}$. Suppose $L$ and $L'$ are of the former type.  Then $\neg [\omega(\texttt{y}) \succeq \omega(\texttt{x})] \in C$, hence we may form $C_1$ by resolving $C$ with the (T.5.a) clause $\neg[\texttt{y} \succeq \texttt{x}] \vee [\omega(\texttt{y}) \succeq \omega(\texttt{x})]$. Then form $C_2$ by resolving $C_1$ with the (T.1) clause $[\texttt{x} \succeq \texttt{y}] \vee [\texttt{y} \succeq \texttt{x}]$, and finally form $C_3$ by resolving $C_2$ with the (T.5.a) clause $\neg[\texttt{x} \succeq \texttt{y}] \vee [\omega'(\texttt{x}) \succeq \omega'(\texttt{y})]$.  Thus:
\[
    C_3 = \big(C \setminus \{L\}\big) \; \cup \; \big\{[\omega'(\texttt{x}) \succeq \omega'(\texttt{y})]\big\}.
\]
Then $C_3$ and $C'$ can be resolved to form $D$.  By \autoref{resolutionlemma}, $\Theta$ and $\Theta \cup \big\{D\big\}$ are logically equivalent. \medskip

Proceeding, suppose now instead that $L$ and $L'$ are of the latter type. Again form $C_1$ via resolving $C$ with the (T.5.a) clause $\neg[\texttt{y} \succeq \texttt{x}] \vee [\omega(\texttt{y}) \succeq \omega(\texttt{x})]$, and then $C_2$ by resolving $C_1$ and the (T.5.b) clause $[\texttt{y} \succeq \texttt{x}] \vee \neg [\omega(\texttt{y}) \succeq \omega(\texttt{x})]$. Finally, form $C_3$ by resolving $C_2$ with the type (T.2.b) $[\omega(\texttt{y}) \succeq \omega(\texttt{x})] \vee [\omega(\texttt{x}) \succ \omega(\texttt{y})]$.  Then $D$ is the resolvent of $C_3$ and $C'$, and hence by an analogous argument, $\Theta$ and $\Theta \cup \big\{D\big\}$ are again logically equivalent.\medskip

Since $C,C'$ and $D$ were arbitrary, we have shown that $\Theta$ and $\Theta \cup \mathcal{F}^1$ are logically equivalent. However, nothing in the preceding argument relied on $C,C'$ belonging to $\mathcal{F}^0$, rather than any other $\mathcal{F}^n$. Hence by an identical argument, $\Theta \cup \mathcal{F}^n$ and $\Theta \cup \mathcal{F}^{n+1}$ are logically equivalent, implying so too are $\Theta$ and $\Theta \cup \mathcal{F}^*$. Since any model evaluates the empty clause $\varnothing$ to $\bot$, the fact $\langle \varnothing, \varnothing\rangle \in \mathcal{F}^*$, implies $\mathcal{F}^*$ is unsatisfiable by soundness of resolution, and hence so too is $\Theta$. Thus no $\mathcal{M}$-invariant weak order extension of $\langle \succsim^R, \succ^R \rangle $ can exist.
\end{proof}

\begin{lemma}\label{structurelemma}
    Let $C$ be a disjunction of negative literals such that $C \in \Phi$ or $C$ is the resolvent of two elements of $\Phi$, one of which contains no positive literals. Then $C \in \mathcal{F}^1$.
\end{lemma}
\begin{proof}
    Suppose first that $C \in \Phi$. Since $C$ is a disjunction of negative literals, it must be of the form (T.2.a), i.e.:
    \[
        C = \neg[\texttt{x} \succeq \texttt{y}] \vee \neg[ \texttt{y} \succ \texttt{x}].
    \]
    Then $C$ corresponds to (the clausal representation of) the forbidden subrelation associated with:
    \[
        \begin{aligned}
            x \succsim^R x \\
            y \succsim^R y,
        \end{aligned}
    \]
    and hence $C \in \mathcal{F}^0 \subseteq \mathcal{F}^1$.  Suppose instead then that $C$ is the resolvent of $C', D \in \Phi$, where $D$ is a disjunction of negative literals and hence $D = \neg[\texttt{x} \succeq \texttt{y}] \vee \neg[ \texttt{y} \succ \texttt{x}]$.  Since $C$ also contains no positive literals, it must be the case that $C' \in \Phi$ contains exactly one positive literal. Therefore it must be of the form (T.3), (T.4) or (T.5).\medskip

    \noindent {\bf Case}: $C' = \neg [\texttt{x} \succeq \texttt{z}] \vee \neg[ \texttt{z} \succeq \texttt{y}] \vee [\texttt{x} \succeq \texttt{y}]$. Then:
    \[
        \begin{aligned}
            x \succsim_c x\\
            z \succsim_c z \\
            y \succsim_c y
        \end{aligned}
    \]
    defines a broken cycle for which which $C$ is a forbidden subrelation, and hence again $C \in \mathcal{F}^0 \subseteq \mathcal{F}^1$.\medskip

    \noindent {\bf Case}: $C' = [\texttt{x} \succeq \texttt{y}]$ or $C' = [\texttt{y} \succ \texttt{x}]$.  If the former is true, then $C = \neg[\texttt{y} \succ \texttt{x}]$. But since $C'$ must be a type (T.4) clause, this implies we must have $x \succsim_c y$ in the data, and hence:
    \[
        \begin{aligned}
            x \succsim_c y
        \end{aligned}
    \]
    is a broken cycle with forbidden subrelation $C = \neg[\texttt{y} \succ \texttt{x}]$ as desired. If instead the latter is true, by an analogous argument $x \succ_c y$ and:
    \[
        \begin{aligned}
            x \succ_c y
        \end{aligned}
    \]
    is a broken cycle which admits forbidden subrelation $C = \neg [\texttt{y} \succeq \texttt{x}]$. In either case, we again find $C \in \mathcal{F}^0 \subseteq \mathcal{F}^1$. \medskip

    {\bf Case}: $C' = [\texttt{x} \succeq \texttt{y}] \vee \neg [\omega(\texttt{x}) \succeq \omega(\texttt{y})] $. Then:
    \[
        C =  \neg [\texttt{y} \succ \texttt{x}] \vee \neg [\omega(\texttt{x}) \succeq \omega(\texttt{y})].
    \]
    However, the broken cycles:
    \[
        \begin{aligned}
            x \succsim^R x & \quad \quad \quad & \omega(x) \succsim^R \omega(x) \\
            y \succsim^R y & \quad \quad \quad & \omega(y) \succsim^R \omega(y)
        \end{aligned}
    \]
    yield forbidden subrelations:
    \[
       \neg [\texttt{x} \succeq \texttt{y}] \vee \neg [\texttt{y} \succ \texttt{x}]
    \]
    and
    \[
         \neg [\omega(\texttt{x}) \succeq \omega(\texttt{y})] \vee \neg [\omega(\texttt{y}) \succ \omega(\texttt{x})].
    \]
    Thus letting $L = [\texttt{x} \succeq \texttt{y}]$ and $L' = [\omega(\texttt{y}) \succ \omega(\texttt{x})]$, $C$ is simply the collapse of these two forbidden subrelations and hence belongs to $\mathcal{F}^1$. An analogous argument obtains if instead $C' = \neg[ x' \succeq y'] \vee [\texttt{x} \succeq \texttt{y}]$ where for some $\omega \in \mathcal{M}$ we have $\omega(x') = x$ and $\omega(y') = y$. 
\end{proof}

\begin{lemma}\label{resolutioncomplete}
    Suppose there does not exist an $\mathcal{M}$-invariant weak order extension of $\langle \succsim^R, \succ^R\rangle$.  Then $\varnothing \in \mathcal{F}^*$.
\end{lemma}
\begin{proof}
    By construction, there is a one-to-one correspondence between $\mathcal{M}$-invariant preference relations extending $\langle \succsim^R, \succ^R\rangle$ and models for $\Phi$.  Thus if no such extension exists, $\Phi$ is unsatisfiable. By Propositional Compactness (see \citealt{schoning2008logic} Chapter I.4), there exists a finite unsatisfiable subset $\Phi^* \subseteq \Phi$. \medskip
    
    By \autoref{completenessnegres}, there exists a derivation of the empty set via negative resolution, i.e. there exists a sequence of clauses $C_1,\ldots, C_N$ such that (i) $C_N = \varnothing$, (ii) for all $1\le n \le N-1$ the clause $C_n$ either belongs to $\Phi^*$ or is the resolvent of two clauses $C_i$ and $C_j$, with $i,j < n$, one of which contains no positive literals. \medskip

    Let $\{D_1,\ldots, D_K\}$ denote those clauses in $\{C_1,\ldots, C_N\}$ which contain no positive literals. For each $D_k$, if $D_k$ is the resolvent of some $C_i$ and $D_j$, define $D_j$ to be its {\bf negative parent} (if $D_k$ is not a resolvent, then we say $D_k$ has no negative parent). Furthermore, if $C_i$ itself is the resolvent of some $C_{i'}$ and $D_{j'}$, then we say $D_{j'}$ is the {\bf negative grandparent} of $D_k$ (similarly, if $C_i \in \Phi^*$, i.e.\ $C_i$ is not a resolvent, then we say $D_k$ has no negative grandparent). Define $\mathcal{NP}(D_k)$, the {\bf negative predecessors} of $D_k$, as the set consisting of $D_k$'s negative parent and grandparent (if these exist). \medskip

    Let $\mathcal{D}^0 \subseteq \{D_1,\ldots, D_K\}$ denote the subset of all $D_k$ which belong to $\mathcal{F}^0$.\footnote{Note $\mathcal{D}^0$ is non-empty as it contains at least $D_1 \in \Phi^*$.} For each $n \ge 1$, define inductively:
    \[
        \mathcal{D}^n = \big\{D_k : \mathcal{NP}(D_k) \subseteq \mathcal{D}^{n-1} \big\} \; \cup \; \mathcal{D}^{n-1}.
    \]
    In other words, $\mathcal{D}^n$ consists of those positive-literal-free clauses $D_k$ all of whose negative predecessors (if these exist) belong to $\mathcal{D}^{n-1}$ or lower. Viewing $\{C_1,\ldots, C_N\}$ as a binary tree (\citealt{schoning2008logic}, Chapter I.5), by \autoref{structurelemma} the sets $\{\mathcal{D}^n\}_{n=0}^\infty$ cover $\{D_1,\ldots, D_K\}$.\footnote{Viewing the resolution proof as a finite binary tree, \autoref{structurelemma} shows that (i) every leaf that belongs to $\{D_1,\ldots, D_K\}$ belongs to $\mathcal{D}^0$, and (ii) every element of $\{D_1,\ldots, D_K\}$ that two leaves for parents belongs to $\mathcal{D}^1$. The claim then follows by inducting on how many generations of ancestors an element of $\{D_1,\ldots, D_K\}$ has in the tree.} We now wish to show that for all $n \ge 1$, $\mathcal{D}^n \subseteq \mathcal{F}^n \subseteq \mathcal{C}^*$.  By definition, $\mathcal{D}^0 \subseteq \mathcal{F}^0$. Thus suppose now that for all $n \le M$, we have $\mathcal{D}^n \subseteq \mathcal{F}^n$, and consider $n = M+1$.  Let $D_k\in \mathcal{D}^{M+1}$. We consider three cases.\medskip

    {\bf Case 1}: $D_k$ has negative parent $D_j$ and negative grandparent $D_{j'}$, both of which belong to $\mathcal{D}^M$ and hence $\mathcal{F}^M$ by the inductive hypothesis.  Then $D_k$ is the resolvent of $D_j$ and some $C_i$, and $C_i$ the resolvent of $D_{j'}$ and some $C_{i'}$. Since $D_k$ and $D_j$ contain no positive literals, this means $C_i$ must contain exactly one positive literal. In turn, since $D_{j'}$ contains no positive literals, this implies $C_{i'}$ must contain exactly two positive literals.  Since $\Phi^*$ contains no clauses with more than two positive literals, and since every resolvent in $\{C_1,\ldots, C_N\}$ has a parent containing no positive literals, no resolvent in $\{C_1,\ldots, C_N\}$ can have more than 1 positive literal. This means that $C_{i'} \in \Phi^*$ and hence is either of the form $C_{i'} = [\texttt{x} \succeq \texttt{y} ] \vee [\texttt{y} \succeq \texttt{x}]$ or $C_{i'} = [\texttt{x} \succeq \texttt{y}] \vee [\texttt{y} \succ \texttt{x}]$. Suppose first that $C_{i'}$ is of the former form. Then $C_i$ consists of $D_{j'}$ but with one literal reversed (i.e. the swapping the positions of the two alternatives featuring in it) and made positive. Since this is $C_i$'s only positive literal, it must be the cancelling literal when it is resolved with $D_j$, thus $D_k$ is precisely the collapse of $D_{j'}$ and $D_{j}$, where the collapse comes from cancelling a pair of reversed weak relations.  If, instead, $C_{i'}$ is of the latter form, then once again $C_i$ consists of $D_{j'}$ but now the one literal is reversed, made positive, and made strict if it was weak, or vice-versa.  This is then cancelled by resolving with $D_j$ and hence $D_k$ consists of the collapse of $D_{j'}$ and $D_j$ where the collapse occurs between weak and strict opposing negative literals. In either case, we find that $D_k$ is the collapse of two elements of $\mathcal{F}^M$ and hence belongs to $\mathcal{F}^{M+1}$ as desired.\medskip

    {\bf Case 2}: $D_k$ has a negative parent $D_j$ but no negative grandparent, i.e.\ $C_i \in \Phi^*$.  Since $D_j$ and $D_k$ contain no positive literals, it must be that $C_i$ contains exactly one positive literal.  Thus $C_i$ is either of the form:
    \begin{itemize}
        \item[(i)] $C_i = \neg [\texttt{x} \succeq \texttt{z}] \vee \neg [\texttt{z} \succeq \texttt{y}] \vee [\texttt{x} \succeq \texttt{y}]$
        \item[(ii)] $C_i = [\texttt{x} \succeq \texttt{y}]$ or $C_i = [\texttt{y} \succ \texttt{x}]$
        \item[(iii)] $C_i = \neg [\texttt{x} \succeq \texttt{y}] \vee [\omega(\texttt{x}) \succeq \omega(\texttt{y})]$ or $C_i = [\texttt{x} \succeq \texttt{y}] \vee \neg [\omega(\texttt{x}) \succeq \omega(\texttt{y})]$ for some $\omega \in \mathcal{M}$.
    \end{itemize}
    Suppose first $C_i$ is of form (i). Then the cancelling literal must be $[\texttt{x} \succeq \texttt{y}]$. However, since:
    \[
        \begin{aligned}
            x \succsim_c x\\
            z \succsim_c z\\
            y \succsim_c y
        \end{aligned}
    \]
    is a broken cycle, we know $\neg[\texttt{x} \succsim \texttt{z}] \vee [\texttt{z} \succsim \texttt{y}] \vee \neg [\texttt{y} \succ \texttt{x}]$ belongs to $\mathcal{F}^0$. Therefore $D_k$ can be formed from collapsing $\neg[\texttt{x} \succeq \texttt{z}] \vee [\texttt{z} \succeq \texttt{y}] \vee \neg [\texttt{y} \succ \texttt{x}] \in \mathcal{F}^0$ with $D_j$. Since $D_j \in \mathcal{F}^M$, this means $D_k \in \mathcal{F}^{M+1}$ as desired.\medskip

    Suppose now that $C_i$ is of type (ii). In the first case, 
    \[
        D_k = D_j \setminus \{\neg [\texttt{x} \succeq \texttt{y}]\}.
    \]
    Note however that if $[\texttt{x} \succeq \texttt{y}] \in \Phi^*$, then $x \succsim_c y$, and thus:
    \[
        x \succsim_c y
    \]
    is a forcing collection for $\neg[\texttt{y} \succ \texttt{x}]$ and hence this clause belongs to $\mathcal{F}^0$. Thus $D_k$ may be obtained as the collapse of $\neg[\texttt{y} \succ \texttt{x}]$ and $D_j$ and hence belongs to $\texttt{D}^{M+1}$. On the other hand, if $C_i$ equals $[\texttt{y} \succ \texttt{x}]$ then $y \succ_c x$ and hence:
    \[
        y \succ_c x
    \]
    is a strict broken cycle for $\neg[\texttt{x} \succeq \texttt{y}]$ and since:
    \[
        D_k = D_j \setminus \{\neg [\texttt{y} \succ \texttt{x}]\},
    \]
    $D_k$ is just the collapse of $D_j$ and $\neg[\texttt{x} \succeq \texttt{y}]$, and hence once again belongs to $\mathcal{D}^{M+1}$.

    Finally, suppose that $C_i$ is of the former type (iii). Then the cancelling literal must be $[\omega(\texttt{x}) \succeq \omega(\texttt{y})]$. Thus $D_k$ is equal to $D_j$ but with the literal $\neg [\omega(\texttt{x}) \succeq \omega(\texttt{y})] \in D_j$ becoming $\neg[\texttt{x} \succeq \texttt{y}] \in D_k$. Now,
    \[
        \begin{aligned}
            x \succsim_c x\\
            y \succsim_c y
        \end{aligned}
    \]
    is a broken cycle hence $\neg [\texttt{x} \succeq \texttt{y}] \vee \neg [\texttt{y} \succ \texttt{x}]$ belongs to $\mathcal{F}^0$. Then $D_k$ arises as the collapse of $D_j\in \mathcal{D}^M$ and $\neg [\texttt{x} \succeq \texttt{y}] \vee \neg [\texttt{y} \succ \texttt{x}] \in \mathcal{F}^0 \subseteq \mathcal{D}^M$ along the pair $\neg [\texttt{y} \succ \texttt{x}]$ and $\neg [\omega(\texttt{x}) \succeq \omega(\texttt{y})]$, and hence belongs to $\mathcal{D}^{M+1}$ as desired. If instead $C_i$ is of the latter type (iii), an analogous argument suffices. \medskip

    {\bf Case 3}: $D_k$ has no negative parent. In this case, $D_k$ cannot be a resolvent at all, and hence belongs to $\Phi^*$.  The only clauses in $\Phi^*$ which contain no positive literals are of the form $\neg[\texttt{x} \succeq \texttt{y}] \vee \neg [\texttt{y} \succ \texttt{x}]$. If $D_k$ is of this form, then it belongs to $\mathcal{F}^0$ as    
    \[
        \begin{aligned}
            x \succsim_c x\\
            y \succsim_c y
        \end{aligned}
    \]
    is a broken cycle for it, and hence it belongs to $\mathcal{F}^{M+1}$ as well.\medskip

Now, as $D_k$ cannot have a negative grandparent without a negative parent (as our proof of inconsistency is by \emph{negative} resolution), these cases are exhaustive, and we find that for all $1 \le k \le K$, the clause $D_k \in \mathcal{F}^*$.  Since $D_K = \varnothing$, this implies that $\varnothing \in \mathcal{F}^*$ as desired.
\end{proof}

\noindent The proof of \autoref{generalchar} follows from these lemmas.   

\section{Proof of \autoref{invariantdm}}

\begin{proof}
    Suppose first that $\big\{[\texttt{y} \succ \texttt{x} ]\big\} \in \mathcal{F}^*$.  By an identical argument to that in the proof of \autoref{generalchar}, $\Phi \cup \big \{ [\texttt{y} \succ \texttt{x} ]\big\}$ is unsatisfiable. Thus no model $\mu$ for $\Phi$ evaluates $\mu \big([\texttt{y} \succ \texttt{x} ]\big) = \top $. Since the set of models for $\Phi$ are in 1-1 correspondence with the set of $\mathcal{M}$-invariant rationalizing preferences of $\langle \succsim^R, \succ^R\rangle$ (which is non-empty by hypothesis), we conclude every such rationalizing preference must weakly rank $x$ above $y$. An identical argument holds for the case of $\big\{[\texttt{y} \succeq \texttt{x}]\big\} \in \mathcal{F}^*$ case.\medskip

    Conversely, suppose every $\mathcal{M}$-invariant rationalizing preference $\succeq^*$ ranks $x \succeq^* y$. Then no model for $\Phi$ evaluates $[\texttt{y} \succ \texttt{x}]$ to $\top$, and hence $\Phi \cup \big \{  [\texttt{y} \succ \texttt{x} ]\big\}$ is unsatisfiable.  Define $\Phi'$ as follows. First, remove from $\Phi$ any clause containing the literal $ [\texttt{y} \succ \texttt{x}]$; then for every remaining clause that contains the negative literal $ \neg [\texttt{y} \succ \texttt{x}]$, delete this literal from it.  By construction, any model $\mu'$ for $\Phi'$ uniquely extends to a model $\mu$ for $\Phi$ which evaluates $\mu\big([\texttt{y} \succ \texttt{x}]\big) = \top$. Since no such models $\mu$ exist, $\Phi'$ must be unsatisfiable. By Propositional Compactness (see \cite{schoning2008logic} Chapter I.4), there exists a finite subset of $\Phi'' \subseteq \Phi'$ that is unsatisfiable; by \autoref{completenessnegres}, there exists a derivation $\{C_1,\ldots, C_N\}$ of $\varnothing$ from $\Phi''$ via negative resolution.  Let $\{D_1,\ldots, D_K\} \subset \{C_1,\ldots, C_N\}$ denote the elements of $\{C_1,\ldots, C_N\}$ belonging to $\Phi''$.  Note that each $D_k$ either (i) belongs to $\Phi$ as well, or (ii) $D_k \cup \{\neg[\texttt{y} \succ \texttt{x}]\}$ belongs to $\Phi$. Moreover, since $\Phi$ is satisfiable by hypothesis, at least one $D_k$ must be of the latter type. Define:
    \[
        \bar{D}_k = \begin{cases} C_i & \textrm{ if } D_k\in \Phi\\ D_k \cup \{\neg [\texttt{y} \succ \texttt{x}]\} & \textrm{ else.}  \end{cases}
    \]
    Then resolving the $\{\bar{D}_1,\ldots, \bar{D}_K\}$ in the same order as in the derivation $\{C_1,\ldots C_N\}$ generates a partial derivation $\{\bar{C}_1, \ldots, \bar{C}_N\}$ of $\neg [\texttt{y} \succ \texttt{x}]$ from $\Phi$ via negative resolution, and hence by an identical argument to \autoref{resolutioncomplete} $[\texttt{y} \succ \texttt{x}] \in \mathcal{F}^*$. An identical argument again works for the case in which every extension ranks $x \succ^* y$.
\end{proof}

\end{appendix}

\pagebreak
\bibliographystyle{ecta}
\bibliography{notes}
\end{document}